\documentclass[11pt]{article}

\usepackage[colorlinks]{hyperref}
\usepackage{amsmath} 
\usepackage{amsthm} 
\usepackage{amssymb}	
\usepackage{graphicx} 
\usepackage{multicol} 
\usepackage{multirow}
\usepackage{color}
\usepackage[dvips,letterpaper,margin=1in,bottom=1in]{geometry}
\usepackage[capitalize,noabbrev]{cleveref}

\usepackage[utf8]{inputenc}
\usepackage[english]{babel}

\usepackage{mathtools}
\usepackage[shortlabels]{enumitem}

\newtheorem{theorem}{Theorem}[section]

\newtheorem{lemma}[theorem]{Lemma}
\newtheorem{corollary}[theorem]{Corollary}

\newtheorem{fact}[theorem]{Fact}

\newtheorem{definition}{Definition}[section]

\newcommand{\braket}[2]{\left< #1 \vphantom{#2} \middle| #2 \vphantom{#1} \right>} 
\newcommand{\ketbra}[2]{\ensuremath{\ket{#1}\!\bra{#2}}}

\DeclarePairedDelimiter\rbra{\lparen}{\rparen}
\DeclarePairedDelimiter\sbra{\lbrack}{\rbrack}
\DeclarePairedDelimiter\cbra{\{}{\}}
\DeclarePairedDelimiter\abs{\lvert}{\rvert}
\DeclarePairedDelimiter\Abs{\lVert}{\rVert}
\DeclarePairedDelimiter\ceil{\lceil}{\rceil}
\DeclarePairedDelimiter\floor{\lfloor}{\rfloor}
\DeclarePairedDelimiter\ket{\lvert}{\rangle}
\DeclarePairedDelimiter\bra{\langle}{\rvert}

\newcommand{\tr} {\operatorname{tr}}
\newcommand{\poly} {\operatorname{poly}}

\newcommand{\sgn} {\operatorname{sgn}}

\newcommand{\Real} {\operatorname{Re}}

\usepackage{algorithm}
\usepackage{algpseudocode}

\usepackage{tabularx}
\usepackage{booktabs}
\usepackage{threeparttable}

\usepackage{adjustbox}

\newcommand{\footremember}[2]{%
    \footnote{#2}
    \newcounter{#1}
    \setcounter{#1}{\value{footnote}}%
}

\usepackage{tikz}
\usetikzlibrary{quantikz2}

\begin{document}

\title{Information-Theoretic Lower Bounds for Approximating Monomials via Optimal Quantum Tsallis Entropy Estimation}
\author{Qisheng Wang \footremember{1}{Qisheng Wang is with the School of Informatics, University of Edinburgh, Edinburgh, United Kingdom (e-mail: \url{QishengWang1994@gmail.com}).}}
\date{}

\maketitle

\begin{abstract}
    This paper reveals a conceptually new connection from information theory to approximation theory via quantum algorithms for entropy estimation. 
    Specifically, we provide an information-theoretic lower bound $\Omega(\sqrt{n})$ on the approximate degree of the monomial $x^n$, compared to the analytic lower bounds shown in \hyperlink{cite.NR76}{Newman and Rivlin (\textit{Aequ.\ Math.}\ 1976)} via Fourier analysis and in \hyperlink{cite.SV14}{Sachdeva and Vishnoi (\textit{Found.\ Trends Theor.\ Comput.\ Sci.}\ 2014)} via the Markov brothers' inequality. 
    This is done by relating the polynomial approximation of monomials to quantum Tsallis entropy estimation. 
    This further implies a quantum algorithm that estimates to within additive error $\varepsilon$ the Tsallis entropy of integer order $q \geq 2$ of an unknown probability distribution $p$ or an unknown quantum state $\rho$, using $\widetilde \Theta(\frac{1}{\sqrt{q}\varepsilon})$ queries to the quantum oracle that produces a sample from $p$ or prepares a copy of $\rho$, improving the prior best $O(\frac{1}{\varepsilon})$ via the Shift test due to \hyperlink{cite.EAO+02}{Ekert, Alves, Oi, Horodecki, Horodecki and Kwek (\textit{Phys.\ Rev.\ Lett.}\ 2002)}. 
    To the best of our knowledge, this is the \textit{first} quantum entropy estimator with optimal query complexity (up to polylogarithmic factors) for all parameters simultaneously. 
\end{abstract}

\newpage
\tableofcontents
\newpage

\section{Introduction}

In approximation theory, the approximate degree of a function $f \colon A \to B$ to precision $\varepsilon$, denoted by $\widetilde{\deg}_\varepsilon\rbra{f, A, B}$, is the minimum degree of a polynomial $p \colon A \to B$ satisfying 
\begin{equation}
    \sup_{x \in A} \abs*{f\rbra{x} - p\rbra{x}} \leq \varepsilon.
\end{equation}
The theory of (best) polynomial approximation was initiated by \cite{Ber14,Ber38} and later developed as an individual field of mathematics (cf.\ \cite{Tim63,Bus12,Tre19}).
Recently, several applications of best polynomial approximations have been found in computer science, e.g., \cite{VV11a,OSV12,SV14,VV17,JVHW15,JVHW17,WY16,AOST17,AA22,LW25}. 

As one of the basic problems in approximation theory, the polynomial approximation of monomials was investigated by Chebyshev in 1858 (see \cite[Page 58]{Akh56}), where he showed that $x^{2n+1}$ can be approximated uniformly on $\sbra{-1, 1}$ by a polynomial of degree at most $2n$ to precision $2^{-2n}$. 
Zolotarev (in 1868, see \cite[Page 281]{Akh56}) and Bernstein (in 1913, see \cite{Ber58}) further investigated the approximation of $x^{2n+1} - \sigma x^{2n}$ by polynomials of degree at most $2n-1$, where $\sigma$ is a given constant. 
In 1930, Akhiezer (see \cite[Page 287]{Akh56}) considered the problem of approximating the monomial $x^n$ on $\sbra{-1, -a}$ and $\sbra{a, 1}$ ($0 \leq a < 1$) by polynomials of degree at most $n-1$. 
From a practical view, the polynomial approximation of monomials was discussed in \cite[Section VII-9]{Lan56} and \cite[Section 4]{Cod70}.
In particular in \cite{NR76} (extending \cite{RJ72} and further extended by \cite{Red78}), it was shown via Fourier analysis that there is a polynomial of degree $O\rbra{\sqrt{n}}$ that approximates $x^n$ uniformly, together with a matching lower bound:

\begin{theorem}[Analytic lower bounds for approximating monomials, adapted from \cite{NR76}]
    For sufficiently small $\varepsilon > 0$ and for any integer $n \geq 2$, 
    \begin{equation}
        \widetilde{\deg}_\varepsilon\rbra{x^n, \sbra{-1, 1}, \mathbb{R}} = \Omega\rbra{\sqrt{n}}.
    \end{equation}
\end{theorem}

Another proof of the lower bound for approximating monomials was noted in \cite{SV14} using the Markov brothers' inequality \cite{Mar90,Mar92}, which further provides an explicit range of $\varepsilon$. 

\begin{theorem} [Analytic lower bounds for approximating monomials, implied by \cite{SV14}, see \cref{thm:lb-by-sv14}] \label{thm:sv14-intro}
    For every constant $\varepsilon \in \rbra{0, \frac{e-1}{4e} \approx 0.1580}$ and sufficiently large integer $n$, 
    \begin{equation}
        \widetilde{\deg}_\varepsilon\rbra{x^n, \sbra{-1, 1}, \mathbb{R}} = \Omega\rbra{\sqrt{n}}.
    \end{equation}
\end{theorem}

It turns out that prior techniques for proving the lower bounds on approximate degrees are usually based on analytic properties of the function to be approximated (other examples include \cite{Tim63,AA22}). 
In this paper, we provide a conceptually new proof of the lower bound for approximating monomials from the perspective of quantum entropy estimation, which reveals a novel connection from information theory to approximation theory via quantum computing. 
This new proof is sketched in \cref{sec:proof-intro}. 
At the key step of this new proof, is a quantum algorithm for Tsallis entropy estimation, whose construction relates to the polynomial approximation of monomials. 
On the other hand, we also provide an information-theoretic quantum lower bound for Tsallis entropy estimation.
Combining the both yields a lower bound for approximating monomials. 

Furthermore, by substituting a known approximation polynomial of monomials into the quantum algorithm used in the above proof, we obtain a specific approach to Tsallis entropy estimation (see \cref{sec:optimal-intro}).
This approach, to the best of our knowledge, turns out to be the \textit{first} quantum entropy estimator with \textit{optimal} query complexity (up to polylogarithmic factors) for all parameters simultaneously.
This quantum algorithm has its own significance in quantum property testing, as entropy is one of the core quantities in information theory, which quantifies the uncertainty of a random variable or the outcome of a random process.

\subsection{A new proof of lower bounds for approximating monomials} \label{sec:proof-intro}

The key component of the proof is a quantum algorithm that closely relates the estimation of Tsallis entropy to the polynomial approximation of monomials. 
Here, for an $N$-dimensional probability distribution $p$, the Tsallis entropy of order $q$ \cite{Tsa88} for $q > 0$ and $q \neq 1$ is defined by
\begin{equation}
    \mathrm{H}_q\rbra{p} = \frac{1}{1-q} \rbra*{ \sum_{j=0}^{N-1} p_j^q - 1 }.
\end{equation}
A quantum algorithm that estimates $\mathrm{H}_q\rbra{p}$ for integer $q \geq 2$ in the purified quantum query access model (see \cref{sec:def-quantum-query-model}) is given as follows. 
\begin{lemma} [\cref{lemma:q-complexity-non-uniform} simplified] \label{lemma:meta-intro}
    Given a quantum query oracle $\mathcal{O}$ for a probability distribution $p$, for every (non-constant) integer $q \geq 2$, there is a (non-uniform) quantum algorithm that estimates $\mathrm{H}_q\rbra{p}$ to within additive error $\varepsilon$ using \begin{equation}
        O\rbra*{\frac{\widetilde{\deg}_{\rbra{q-1}\varepsilon/2}\rbra{x^{q-1}, \sbra{-1, 1}, \sbra{-1, 1}}}{q\varepsilon}}
    \end{equation}
    queries to $\mathcal{O}$.\footnote{Here, as a technical detail, the best approximation polynomial of $x^{q-1}$ of degree $\widetilde{\deg}_{\rbra{q-1}\varepsilon/2}\rbra{x^{q-1}, \sbra{-1, 1}, \sbra{-1, 1}}$ is not known to be computable. 
    For this reason, the quantum algorithm in \cref{lemma:meta-intro} is not (Turing) uniform. 
    Nevertheless, for any approximation polynomial $p \colon \sbra{-1, 1} \to \sbra{-1, 1}$ of $x^{q-1}$ to precision $\Theta\rbra{q\varepsilon}$ that is efficiently computable, there is a polynomial-time uniform quantum algorithm that estimates $\mathrm{H}_q\rbra{p}$ to within additive error $\Theta\rbra{\varepsilon}$ with query complexity $O\rbra{\frac{\deg\rbra{p}}{q\varepsilon}}$, as shown in \cref{thm:q-tsallis-estimator}.
    For the general case, see \cref{lemma:meta} for details.}
\end{lemma}

The technical idea of the proof of \cref{lemma:meta-intro} turns out to be simple, which was employed in the quantum algorithms for estimating fidelity \cite{GP22}, trace distance \cite{WZ24b,LGLW23}, von Neumann entropy \cite{LGLW23,WZ24}, and Tsallis entropy \cite{LW25}. 
The difference is that the proof of \cref{lemma:meta-intro} uses a polynomial approximation of monomials $x^q$ of degree dependent on $q$, whereas prior work focused on the dependence on the precision. 

To complete the proof, we also provide an information-theoretic lower bound on the quantum query complexity of Tsallis entropy estimation. 

\begin{lemma} [Lower bounds on the quantum query complexity of Tsallis entropy estimation, \cref{thm:lb-summary}] \label{lemma:low-main}
    Given a quantum query oracle $\mathcal{O}$ for a probability distribution $p$, for every (non-constant) integer $q \geq 2$, any quantum query algorithm that estimates $\mathrm{H}_q\rbra{p}$ to within additive error $\varepsilon$ uses $\Omega\rbra{\frac{1}{\sqrt{q}\varepsilon}}$ queries to $\mathcal{O}$.
\end{lemma}
\begin{proof}[Proof sketch of \cref{lemma:low-main}]
    To show a matching lower bound on the quantum query complexity of Tsallis entropy estimation, we consider the problem of distinguishing the two $2$-dimensional probability distributions $p^\pm$:
\begin{equation}
    p^\pm_0 = 1 - \frac{1}{q} \pm \varepsilon, \qquad p^\pm_1 = \frac{1}{q} \mp \varepsilon.
\end{equation}
It can be seen that the difference between the Tsallis entropy of integer order $q \geq 2$ is large enough $\mathrm{H}_q\rbra{p^-} - \mathrm{H}_q\rbra{p^+} \geq \Omega\rbra{\varepsilon}$ so that any algorithm that estimates $\mathrm{H}_q\rbra{p}$ to within additive error $\Theta\rbra{\varepsilon}$ can be used to distinguish the two probability distributions $p^\pm$. 
This hard instance was initially used in \cite{CWYZ25} for considering the sample complexity of distinguishing two quantum states with their spectra being the probability distributions $p^{\pm}$, respectively. 

In our case, to give a lower bound on the quantum query complexity, we use the result in \cite{Bel19} (\cref{thm:qlower-dis-prob}) that distinguishing any two probability distributions $p$ and $p'$ requires quantum query complexity $\Omega\rbra{1/d_{\mathrm{H}}\rbra{p, p'}}$, where $d_{\mathrm{H}}\rbra{p, p'}$ is the Hellinger distance between $p$ and $p'$. 
Then, a lower bound of $\Omega\rbra{\frac{1}{\sqrt{q}\varepsilon}}$ on the quantum query complexity of Tsallis entropy estimation can be obtained by verifying that $d_{\mathrm{H}}\rbra{p^+, p^-} \leq O\rbra{\sqrt{q}\varepsilon}$. 
Similar ideas were used in \cite{LWL24,LWWZ24,Wan24} to show matching quantum lower bounds.
In \cite[Theorem 5.7]{LW25}, they also used the similar idea but with non-optimal probability distributions, thereby only resulting in a lower bound of $O\rbra{\frac{1}{\sqrt{\varepsilon}}}$. 
\end{proof}

By \cref{lemma:meta-intro,lemma:low-main}, we immediately have
\begin{equation}
    O\rbra*{\frac{\widetilde{\deg}_{\rbra{q-1}\varepsilon/2}\rbra{x^{q-1}, \sbra{-1, 1}, \sbra{-1, 1}}}{q\varepsilon}} \geq \Omega\rbra*{\frac{1}{\sqrt{q}\varepsilon}},
\end{equation}
which gives $\widetilde{\deg}_{\rbra{q-1}\varepsilon/2}\rbra{x^{q-1}, \sbra{-1, 1}, \sbra{-1, 1}} \geq \Omega\rbra{\sqrt{q}}$.
With a detailed analysis of the range of $\varepsilon$, we have the following information-theoretic lower bound for approximating monomials, in contrast to the analytic lower bound in \cref{thm:sv14-intro}. 

\begin{theorem} [Information-theoretic lower bounds for approximating monomials, \cref{thm:approx-deg-lb}] \label{thm:xq-lb}
    For every constant $\varepsilon \in \rbra{0, \frac{1}{2e} \approx 0.1839}$ and sufficiently large integer $n$, 
    \begin{equation}
        \widetilde{\deg}_\varepsilon\rbra{x^n, \sbra{-1, 1}, \mathbb{R}} = \Omega\rbra{\sqrt{n}}.
    \end{equation}
\end{theorem}

The proof of \cref{thm:xq-lb} (actually a new proof of \cref{thm:sv14-intro}) reveals a strong connection from quantum computing and information theory to approximation theory. 
This proof is based on the discrimination of probability distributions, thereby establishing an information-theoretic lower bound $\Omega\rbra{\sqrt{n}}$ on the approximate degree of the monomial $x^n$. 
Also note that the range of the parameter $\varepsilon$ in \cref{thm:xq-lb} is slightly wider than that in \cref{thm:sv14-intro}.
In the following, we discuss the conceptual and technical significance of this new proof from the perspectives of both information theory and quantum computing. 

\begin{itemize}
    \item \textbf{From information theory to approximation theory.} 
    Prior work in approximation theory mainly focuses on the dependence of approximate degrees on the precision $\varepsilon$ (cf.\ \cite{Tim63,Bus12}). 
    The dependence on other parameters were recently popularized in computer science, e.g., the exponents $\beta$ and $n$ respectively of $e^{\beta x}$ and $x^n$ in \cite{OSV12,SV14,AA22}, where their methods are based on the analytic properties of derivatives and Chebyshev polynomials. 
    In sharp contrast, our new proof is based on the information-theoretic properties of distinguishing probability distributions (\cref{lemma:low-main}).
    This suggests a novel idea for proving lower bounds on approximate degrees from the perspective of information theory. 
    
    \item \textbf{From quantum computing to approximation theory.} 
    Prior work in quantum computing uses polynomial approximations to derive quantum lower bounds \cite{MS24} as well as quantum upper bounds (via quantum singular value transformation \cite{GSLW19}). 
    In sharp contrast, this work does the reverse: show lower bounds on approximate degrees by quantum upper and lower bounds (\cref{lemma:meta-intro,lemma:low-main}). 
    Some readers may wonder if the information-theoretic proof can be made quantum-independent. 
    To our knowledge, quantum computing seems to be necessary in this proof, as there is no known similar tight relation between Tsallis entropy estimation and best polynomial approximation. 
    The (prior best) classical estimator for Tsallis entropy of non-constant order $q \geq 2$ based on best polynomial approximations is due to \cite[Theorem 1]{JVHW17} with sample complexity $O\rbra{\frac{1}{\varepsilon^2}}$, whereas only a lower bound of $\Omega\rbra{\frac{1}{q^3\varepsilon^2}}$ is known in \cite[Section VI-A]{JVHW17}.
\end{itemize}

In addition, the proof of \cref{thm:xq-lb} shows a lower bound on approximate degrees from the quantum upper bounds for entropy estimation, which is similar (in spirit) to the work \cite{BDBGK18} where they
proved classical lower bounds from quantum upper bounds.
In sharp contrast, existing results usually focus on the converse direction: proving quantum lower bounds from lower bounds on approximate degrees, e.g., \cite{BBC+01,Kut05,BKT20,BT22,SY23,MS24}. 

\subsection{An optimal quantum estimator for Tsallis entropy} \label{sec:optimal-intro}

In this section, we will provide a quantum estimator for the Tsallis entropy of integer order $q \geq 2$ of a quantum state $\rho$:
\begin{equation}
    \mathrm{S}_q\rbra{\rho} = \frac{1}{1-q} \rbra[\big]{\tr\rbra{\rho^q} - 1},
\end{equation}
in the purified quantum query access model (see \cref{sec:def-quantum-query-model}), which also implies a quantum estimator for probability distributions. 

In addition to its theoretical importance alone, the estimation of Tsallis entropy of integer order is used as a crucial subroutine in many applications, e.g., computing nonlinear functionals of quantum states \cite{EAO+02,Bru04,BCE+05,vEB12,ZL24,CWYZ25} and entanglement spectroscopy \cite{JST17,SCC19,YS21}. 
Efficient approaches were proposed in the literature. 
The first approach is via the Shift test \cite{EAO+02}, generalizing the SWAP test \cite{BCWdW01}. 
Recently, an approach with constant quantum depth was given in \cite{QKW24} and an approach for the low-rank case was given in \cite{SLLJ24}. 

In the following, we first review the techniques of the prior approach of \cite{EAO+02}, and then present our quantum estimator for Tsallis entropy that is significantly more efficient. 

\paragraph{Warm-up: Tsallis entropy estimation via the Shift test.} \label{sec:warmup}

The prior approach based on the Shift test \cite{EAO+02} and quantum amplitude estimation \cite{BHMT02} turns out to have query complexity $O\rbra{\frac{1}{\varepsilon}}$. 
As a warm-up, we briefly introduce how to estimate the Tsallis entropy (of a quantum state $\rho$) via the Shift test \cite{EAO+02} and then analyze the complexity. 
The Shift test allows us to use $q$ samples of a quantum state $\rho$ to produce a random variable $x$ that relates to $\mathrm{S}_q\rbra{\rho}$. 
Specifically, let $W$ be the main part of the Shift test as shown in \cref{fig:shift}, then 
\begin{equation}
    \Pr\sbra{x = 0} = \tr\rbra*{\Pi_0 W \rbra*{\ketbra{0}{0} \otimes \rho^{\otimes q}} W^\dag} = \frac{1 + \tr\rbra{\rho^q}}{2} = 1 + \frac{1-q}{2} \mathrm{S}_q\rbra{\rho},
\end{equation}
where $\Pi_0 = \ketbra{0}{0} \otimes I$ is a projector onto the subspace with the first qubit being $\ket{0}$ and $\textup{Shift}_q$ is the unitary operator that shifts $q$ quantum states: $\textup{Shift}_q \ket{\phi_1} \ket{\phi_2} \dots \ket{\phi_q} = \ket{\phi_q} \ket{\phi_1} \dots \ket{\phi_{q-1}}$. 

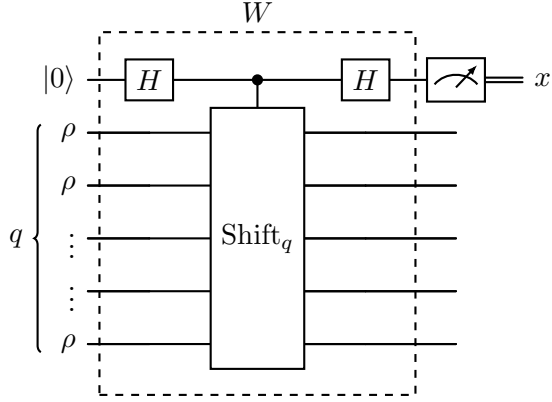
\begin{figure} [!htp]
    \centering
    \begin{quantikz} [row sep = {20pt, between origins}]
        & \lstick{$\ket{0}$} \wireoverride{n} & \gate{H} \gategroup[6,steps=3,style={dashed,inner sep=6pt}]{$W$} & \ctrl{2} & \gate{H} & \meter{} & \setwiretype{c} \rstick{$x$} \\
        \lstick[5]{$q$} & \lstick{$\rho$} \wireoverride{n} & \qw & \gate[5]{\textup{Shift}_q} & \qw & \qw \\
        & \lstick{$\rho$} \wireoverride{n} & \qw & & \qw & \qw \\
        & \lstick{$\vdots$} \wireoverride{n} & \qw & & \qw & \qw \\
        & \lstick{$\vdots$} \wireoverride{n} & \qw & & \qw & \qw \\
        & \lstick{$\rho$} \wireoverride{n} & \qw & & \qw & \qw 
    \end{quantikz}
    \caption{Quantum circuit for the Shift test.}
    \label{fig:shift}
\end{figure}

Given a state-preparation circuit $\mathcal{O}$ that prepares a purification of $\rho$ (defined by \cref{eq:def-O-rho}), we can first prepare $q$ samples of $\rho$ by applying $\mathcal{O}^{\otimes q}$ ($q$ queries to $\mathcal{O}$) and then perform the Shift test.
Let $U = W \cdot \mathcal{O}^{\otimes q}$ be the quantum circuit of this process (for simplicity, here we omit which qubits each unitary operator acts on). 
Then, $U$ has the form:
\begin{equation}
    U \ket{0} = \sqrt{\Pr\sbra{x = 0}} \ket{0}\ket{\varphi_0} + \sqrt{1 - \Pr\sbra{x = 0}} \ket{1}\ket{\varphi_1},
\end{equation}
where $\ket{\varphi_0}$ and $\ket{\varphi_1}$ are two unimportant (normalized) pure states. 
To obtain an estimate of $\mathrm{S}_q\rbra{\rho}$ to within additive error $\varepsilon$, we only have to estimate $\Pr\sbra{x = 0}$ to within additive error $\Theta\rbra{q\varepsilon}$, which can be done via quantum amplitude estimation \cite{BHMT02} using $O\rbra{\frac{1}{q\varepsilon}}$ queries to $U$, thus $q \cdot O\rbra{\frac{1}{q\varepsilon}} = O\rbra{\frac{1}{\varepsilon}}$ queries to $\mathcal{O}$. 
Therefore, we obtain a simple estimator for the Tsallis entropy of integer order $q$ of the quantum state $\rho$. 

\begin{fact} [Folklore quantum estimator for Tsallis entropy of integer order, implied by \cite{EAO+02,BHMT02}] \label{fact:folklore}
    Given a unitary oracle $\mathcal{O}$ that prepares a quantum state $\rho$, for (non-constant) integer $q \geq 2$, $\mathrm{S}_q\rbra{\rho}$ can be estimated to within additive error $\varepsilon$ using $O\rbra{\frac{1}{\varepsilon}}$ queries to $\mathcal{O}$. 
\end{fact}

The quantum estimator for Tsallis entropy in \cref{fact:folklore} is not known to be optimal, and only a lower bound of $\Omega\rbra{\frac{1}{\sqrt{\varepsilon}}}$ was given in \cite{LW25} for every constant $q > 1$. 
Given its simplicity and conciseness, one might think that this folklore approach should be optimal. 
Surprisingly, however, it turns out that we can do even better. 

\paragraph{Optimal quantum Tsallis entropy estimator of integer order.}

By extending \cref{lemma:meta-intro}, we obtain a (uniform) quantum algorithm for estimating the Tsallis entropy of integer order of a quantum state in the purified quantum query access model (\cref{sec:def-quantum-query-model}). 

\begin{theorem} [Optimal quantum estimator for Tsallis entropy of integer order of quantum states, \cref{thm:q-tsallis-estimator}] \label{thm:up-main}
    Given a state-preparation circuit $\mathcal{O}$ of a quantum state $\rho$, for every (non-constant) integer $q \geq 2$, the Tsallis entropy $\mathrm{S}_q\rbra{\rho}$ can be estimated to within additive error $\varepsilon$ using ${O}\rbra{\frac{\sqrt{\log\rbra{1/q\varepsilon}}}{\sqrt{q}\varepsilon}}$ queries to $\mathcal{O}$.
\end{theorem}
\begin{proof}[Proof sketch of \cref{thm:up-main}]
    The proof can be done with \cref{lemma:meta-intro} by noting that 
    \begin{equation}
        \widetilde{\deg}_{\rbra{q-1}\varepsilon/2}\rbra{x^{q-1}, \sbra{-1, 1}, \sbra{-1, 1}} = O\rbra*{\sqrt{q\log\rbra*{\frac{1}{q\varepsilon}}}},
    \end{equation}
    using the (efficiently-computable) approximation polynomial of the monomial $x^q$ in \cite{SV14} (\cref{thm:poly-approx-monomial}).\footnote{The quantum algorithm can be made uniform with some subtle modifications (see \cref{thm:q-tsallis-estimator} for the details).}
\end{proof}

It is reasonable to require in \cref{thm:up-main} that $\varepsilon \in (0, \frac{1}{q-1}]$ as $0 \leq \mathrm{S}_q\rbra{\rho} \leq \frac{1}{q-1}$. 
The folklore approach mentioned in \cref{fact:folklore} has quantum query complexity $O\rbra{\frac{1}{\varepsilon}}$. 
In comparison, our quantum estimator for Tsallis entropy in \cref{thm:up-main} further improves the query complexity to $\widetilde{O}\rbra{\frac{1}{\sqrt{q}\varepsilon}}$. 
For example, when $\varepsilon = \frac{1}{100q}$, the query complexity in \cref{fact:folklore} is $\Theta\rbra{q}$ while \cref{thm:up-main} gives $\Theta\rbra{\sqrt{q}}$, which yields a notable speedup. 

As a special case, \cref{thm:up-main} also implies a quantum estimator for the Tsallis entropy of probability distributions, as a probability distribution can be understood as a quantum state with the computational basis being its eigenvectors.

\begin{corollary} [Optimal quantum estimator for Tsallis entropy of integer order of probability distributions] \label{coro:up-main}
    Given a quantum query oracle $\mathcal{O}$ for a probability distribution $p$, for every (non-constant) integer $q \geq 2$, the Tsallis entropy $\mathrm{H}_q\rbra{p}$ can be estimated to within additive error $\varepsilon$ using ${O}\rbra{\frac{\sqrt{\log\rbra{1/q\varepsilon}}}{\sqrt{q}\varepsilon}}$ queries to $\mathcal{O}$.
\end{corollary}

The lower bound $\Omega\rbra{\frac{1}{\sqrt{q}\varepsilon}}$ in \cref{lemma:low-main} reveals that our quantum estimators in \cref{thm:up-main,coro:up-main} are optimal only up to a factor of $O\rbra{\sqrt{\log\rbra{1/q\varepsilon}}} \leq O\rbra{\sqrt{\log\rbra{1/\varepsilon}}}$, where the factor is independent of $q$. 
This means that the dependence on $q$ in \cref{thm:up-main} is optimal only up to a constant factor. 

To the best of our knowledge, this is the \textit{first} quantum entropy estimators that achieve \textit{optimal} query complexity for all parameters simultaneously. 
Prior to this work, the only quantum entropy estimator with query complexity close to optimal in the dependence on the dimension $N$ is the ones for Shannon entropy, where the upper bound is $\widetilde{O}\rbra{\frac{\sqrt{N}}{\varepsilon^2}}$ due to \cite[Theorem 6]{LW19} and later improved to $\widetilde{O}\rbra{\frac{\sqrt{N}}{\varepsilon^{1.5}}}$ in \cite[Theorem 12]{GL20} and the lower bound is $\widetilde{\Omega}\rbra{\sqrt{N}}$ given in \cite{BKT20}, whereas the dependence on $\varepsilon$ remains not optimal yet. 
For the readers who are interested in entropy estimation, we provide an extensive review of both classical and quantum entropy estimators in \cref{sec:review-entropy-estimation}.

\subsection{Related work}

\paragraph{Concurrent work.}
Very recently, the concurrent work \cite{ZWZY25} provided a quantum algorithm for estimating $\tr\rbra{\mathcal{O} \rho^q}$ to within additive error $\varepsilon$ with query complexity $O\rbra{\sqrt{q\log\rbra{\Abs{\mathcal{O}}/\varepsilon}}\Abs{\mathcal{O}}/\varepsilon}$, where $\mathcal{O}$ is an observable, which implies the same query complexity stated in \cref{thm:up-main} by taking $\mathcal{O} = I$. 

\paragraph{Approximation of monomials.}
In addition to polynomial approximations of monomials, rational approximations of monomials have also been investigated in the literature \cite{NR76b,Gan79,New79,NR80,Sta03,NT18}, as well as its approximations by reciprocals of polynomials \cite{Bor81}. 
See \cite{Red87} for more details. 

\paragraph{Polynomial approximation for computational complexity.}
Polynomial approximations of the exponential functions $e^x$ and $e^{-x}$ and their applications were studied in \cite{HL97,OSV12,SV14,AA22}. 
Polynomial approximations of the reciprocal function $\frac{1}{x}$ were studied in \cite{CKS17} for solving systems of linear equations, improving the previous quantum algorithms in \cite{HHL09,Amb12}. 
Polynomial approximations of the sign function $\sgn\rbra{x}$ were studied in \cite{EY07}, and were later used in \cite{LC17} for Hamiltonian simulation, improving the prior works \cite{ATS03,BACS07,CW12,BCC+14,BCC+15,BCK15}. 
Polynomial approximations of the (fractional) power functions $x^{\pm c}$ were studied in \cite{CGJ19,Gil19} for implementing block-encodings of the powers of Hamiltonians.
Polynomial approximations of (piece-wise) smooth functions were studied in \cite{vAGGdW20,GSLW19} with applications in efficiently implementing quantum singular value transformation. 

\subsection{Discussion}

In this paper, we presented a quantum query algorithm for estimating the Tsallis entropy of integer order $q \geq 2$ of both probability distributions and quantum states, which, to the best of our knowledge, is the first quantum entropy estimator with optimal query complexity only up to polylogarithmic factors for all parameters simultaneously. 
This discovery further leads to a conceptually new proof of the lower bound on the approximate degree of monomials from the perspective of quantum computing and information theory. 

We therefore raise some open questions for future research. 

\begin{itemize}
    \item Can we generalize our method to obtain other quantum entropy estimators with optimal query complexity? 
    The most possible candidate can be the estimation of Shannon entropy, as the optimal dependence on $\widetilde{\Theta}\rbra{\sqrt{N}}$ is known due to \cite{LW19,GL20,BKT20}. 
    Can we show a matching lower bound on the quantum query complexity in $\varepsilon$ for Shannon entropy estimation? 
    \item The proof of the lower bound on the approximate degree of monomials given in this paper suggests a new approach to approximation theory from information theory. 
    Can we find other such connections to or even prove new bounds for approximation theory from the perspective of information theory?
    \item In \cite{LW25}, they showed that estimating the Tsallis entropy of integer order $q = 2$ of quantum states is $\mathsf{BQP}$-complete, but it is still unknown for the general case of integer $q \geq 3$. 
    An interesting problem is under what condition of $q$, estimating the Tsallis entropy of order $q$ is $\mathsf{BQP}$-complete?
    Also, can we show the computational hardness of R\'enyi entropy estimation?
\end{itemize}

\section{Review of Entropy Estimation} \label{sec:review-entropy-estimation}

In this section, we provide an extensive review of entropy estimation. 
Classical entropy estimators are collected in \cref{sec:classical} while quantum entropy estimators are collected in \cref{sec:quantum}.
Especially in \cref{sec:q-quantum-entropy}, a summary of the quantum query complexity of entropy estimation is given, confirming that there have been no quantum estimators that achieve a query complexity optimal for all parameters simultaneously, prior to our results (\cref{thm:up-main,coro:up-main}). 

\subsection{Classical entropy estimation} \label{sec:classical}

\paragraph{Shannon entropy.}
For a (discrete) probability distribution $p$ with the sample space of size $N$, the Shannon entropy \cite{Sha48a,Sha48b} of $p$ is defined by
\begin{equation}
    \mathrm{H}\rbra{p} = -\sum_{j=0}^{N-1} p_j \ln\rbra{p_j}.
\end{equation}
The estimation of Shannon entropy has a wide range of applications, e.g., measuring the variability of genetic sequences \cite{SEM91}, analyzing neural data \cite{Pan03,NBvS04}, and detecting network anomalies \cite{LSO+06}. 
From the perspective of algorithms and complexity, a series of work \cite{Pan04,VV11a,VV11b,VV17,JVHW15,JVHW17,WY16} investigated the sample complexity of Shannon entropy estimation, i.e., how many samples drawn from the probability distribution $p$ are sufficient and necessary to estimate the Shannon entropy to within additive error $\varepsilon$ (with high probability, say $\geq \frac{2}{3}$), which was finally shown to be $\Theta\rbra{\frac{N}{\varepsilon\log\rbra{N}}+\frac{\log^2\rbra{N}}{\varepsilon^2}}$ in \cite{JVHW15,WY16}.
In addition, the sample complexity of Shannon entropy estimation with multiplicative error $\gamma$ is known to be $\widetilde{O}\rbra{N^{\rbra{1+\eta}/\gamma^2}}$ provided that $\mathrm{H}\rbra{p} = \Omega\rbra{\frac{\gamma}{\eta}}$ \cite{BDKR05}, and an almost tight lower bound of $\Omega\rbra{N^{\rbra{1-o\rbra{1}}/\gamma^2}}$ was shown in \cite{Val11}. 

\paragraph{Tsallis entropy and R\'enyi entropy.}
A popular generalization of the Shannon entropy is the Tsallis entropy \cite{Tsa88} of order $q$, defined for $q > 0$ and $q \neq 1$ by
\begin{equation}
    \mathrm{H}_q\rbra{p} = \frac{1}{1-q} \rbra*{ \sum_{j=0}^{N-1} p_j^q - 1 },
\end{equation}
which reproduces the Shannon entropy in the limiting case when $q$ approaches to $1$, i.e.,
\begin{equation}
    \mathrm{H}\rbra{p} = \lim_{q \to 1} \mathrm{H}_q\rbra{p}. 
\end{equation}
For the special case of $q = 2$, $\mathrm{H}_2\rbra{p}$ is also known as the Gini impurity \cite{Gin12}, which is now widely used in machine learning \cite{BFOS84}. 
For the general case, the Tsallis entropy has been found to be useful particularly in nonextensive statistical mechanics (cf.\ \cite{LT98,Tsa01,AdSM+10}). 
Tsallis entropy is also closely related to R\'enyi entropy \cite{Ren61}:
\begin{equation}
    \mathrm{H}_q^{\mathrm{R}}\rbra{p} = \frac{1}{1-q} \ln\rbra*{ \sum_{j=0}^{N-1} p_j^q } = \frac{1}{1-q}\ln\rbra[\big]{\rbra{1-q}\mathrm{H}_q\rbra{p} + 1}.
\end{equation}
The sample complexity of Tsallis entropy estimation was investigated in \cite{AK01,JVHW15,JVHW17}, where they divided the cases into several ranges of $q$ with almost tight bounds:
\begin{itemize}
    \item For every constant $0 < q \leq \frac{1}{2}$, an upper bound of $O\rbra{\frac{N^{1/q}}{\varepsilon^{1/q}\log\rbra{N/\varepsilon}}}$ is given, together with a matching lower bound of $\Omega\rbra{\frac{N^{1/q}}{\log\rbra{N}}}$ for constant $\varepsilon$.
    \item For every constant $\frac{1}{2} < q < 1$, an upper bound of $O\rbra{\frac{N^{1/q}}{\varepsilon^{1/q}\log\rbra{N/\varepsilon}}+\frac{N^{2\rbra{1-q}}}{\varepsilon^2}}$ is given when $\varepsilon \geq \frac{1}{\poly\rbra{N}}$, together with a matching lower bound of $\Omega\rbra{\frac{N^{1/q}}{\log\rbra{N}}}$ for constant $\varepsilon$.
    \item For every constant $1 < q < \frac{3}{2}$, the sample complexity is $\Theta\rbra{\frac{1}{\varepsilon^{1/\rbra{q-1}}\log\rbra{1/\varepsilon}}}$. 
    \item For every constant $q \geq \frac{3}{2}$, the sample complexity is $\Theta\rbra{\frac{1}{\varepsilon^2}}$. 
\end{itemize}
It is worth noting that for every constant $q > 1$, the sample complexity of estimating the Tsallis entropy of order $q$ is independent of the dimension $N$, which exhibits an exponential advantage in sharp contrast to the estimation of Shannon entropy and R\'enyi entropy. 
In comparison, the sample complexity of R\'enyi entropy estimation was investigated in \cite{AOST17}, where they showed that:
\begin{itemize}
    \item For constant integer $q \geq 2$, the sample complexity is $\Theta\rbra{\frac{N^{\rbra{q-1}/q}}{\varepsilon^2}}$. 
    \item For constant non-integer $q > 1$, an upper bound of $O\rbra{\frac{N}{\log\rbra{N}}}$ is given for constant additive error, together with an almost matching lower bound of $\Omega\rbra{N^{1-o\rbra{1}}}$.\footnote{We avoid discussing the dependence on the additive error $\varepsilon$ for R\'enyi entropy estimation of non-integer order $q$ as a suspicious error about the dependence on $\varepsilon$ in \cite{AOST17} was pointed out in \cite[Section V-C]{WZL24}. Nevertheless, it was also pointed out in \cite{WZL24} that for every constant $0 < q < 1$, the sample complexity of R\'enyi entropy estimation of order $q$ is upper bounded by the sample complexity of Tsallis entropy estimation of order $q$ in \cite{JVHW15}.}
\end{itemize}

\subsection{Quantum entropy estimation} \label{sec:quantum}

With regard to its importance, entropy is also an important quantity in quantum information. 
For a mixed quantum state $\rho$, the von Neumann entropy \cite{vN27,vN32}, which generalizes the Shannon entropy, is defined by 
\begin{equation}
    \mathrm{S}\rbra{\rho} = -\tr\rbra{\rho\ln\rbra{\rho}}. 
\end{equation}
Similarly, the quantum Tsallis entropy is defined by
\begin{equation}
    \mathrm{S}_q\rbra{\rho} = \frac{1}{1-q} \rbra[\big]{\tr\rbra{\rho^q} - 1},
\end{equation}
and the quantum R\'enyi entropy is defined by
\begin{equation}
    \mathrm{S}_q^{\mathrm{R}}\rbra{\rho} = \frac{1}{1-q} \ln\rbra[\big]{\tr\rbra{\rho^q}}.
\end{equation}
The classical entropy can be viewed as a special case of the quantum entropy by representing the discrete probability distribution $p$ as a mixed quantum state
\begin{equation}
    \rho_p = \sum_{j=0}^{N-1} p_j \ketbra{j}{j},
\end{equation}
where $\cbra{\ket{j}}$ is the computational basis and measuring $\rho$ in the computational basis will give the outcome $j$ with probability $p_j$. 
Then, the entropy of the probability distribution $p$ is equal to the entropy of the corresponding quantum state $\rho_p$: $\mathrm{H}\rbra{p} = \mathrm{S}\rbra{\rho_p}$, $\mathrm{H}_q\rbra{p} = \mathrm{S}_q\rbra{\rho_p}$, and $\mathrm{H}_q^{\mathrm{R}}\rbra{p} = \mathrm{S}_q^{\mathrm{R}}\rbra{\rho_p}$. 
In addition to the aforementioned applications, the entropy estimation of quantum states have applications in quantifying quantum entanglement \cite{FIK08,HGKM10,IMP+15}, preparing quantum Gibbs states \cite{WH19,CLW20,WLW21}, and learning Hamiltonians \cite{AAKS21}. 

\subsubsection{Sample complexity}

Given independent and identical samples of an $N$-dimensional quantum state $\rho$, the sample complexity of estimating the entropy of $\rho$ is measured by the number of independent and identical samples of $\rho$, which is a direct analog of the sample complexity of classical entropy estimation. 
This quantum input model was adopted in, for example, quantum state tomography \cite{HHJ+17,OW16,OW17}.

\paragraph{Von Neumann entropy.}
It was shown in \cite{AISW20} that the von Neumann entropy of a quantum state $\rho$ can be estimated using $O\rbra{\frac{N^2}{\varepsilon^2}}$ samples of $\rho$, which was later improved to $O\rbra{\frac{N^2}{\varepsilon}+\frac{\log^2\rbra{N}}{\varepsilon^2}}$ in \cite{BMW16} (also noted in \cite[Theorem 1.27]{OW17}), whereas the current best lower bound is $\Omega\rbra{\frac{N}{\varepsilon}}$ due to \cite{WZ24} (which is based on the lower bound for the mixedness testing of quantum states given in \cite{OW21}).
In addition, a time-efficient quantum algorithm with sample and time complexity $\widetilde{O}\rbra{\frac{N^2}{\varepsilon^5}}$ was provided in \cite{WZ24}.

\paragraph{R\'enyi entropy.}
The sample complexity of R\'enyi entropy estimation of order $q$ of quantum states was shown to be $O\rbra{\frac{N^{2/q}}{\varepsilon^{2/q}}}$ for every constant $0 < q < 1$ and $O\rbra{\frac{N^2}{\varepsilon^2}}$ for every constant $q > 1$ in \cite{AISW20}, whereas the current best lower bound is $\Omega\rbra{\frac{N}{\varepsilon}+\frac{N^{1/q-1}}{\varepsilon^{1/q}}}$ shown in \cite{WZ24}. 
In addition, a time-efficient quantum algorithm was given in \cite{WZ24} with sample and time complexity $\widetilde{O}\rbra{\frac{N^{4/q-2}}{\varepsilon^{1+4/q}}}$ for every constant $0 < q < 1$ and $\widetilde{O}\rbra{\frac{N^{4-2/q}}{\varepsilon^{3+2/q}}}$ for every constant $q > 1$. 

\paragraph{Tsallis entropy.}
For constant integer $q \geq 2$, the Tsallis entropy of order $q$ of quantum states can be estimated with sample complexity $O\rbra{\frac{1}{\varepsilon^2}}$ via the SWAP test \cite{BCWdW01} (for $q = 2$) and the Shift test \cite{EAO+02} (for $q \geq 2$).
For every constant $q > 1$, the sample complexity was shown to be $\widetilde{O}\rbra{\frac{1}{\varepsilon^{3+2/\rbra{q-1}}}}$ in \cite{LW25}, which is independent of $N$, achieving an exponential improvement over the sample complexity implied by the prior works \cite{AISW20,WGL+24,WZL24,WZ24}.
This was later improved to $\widetilde{O}\rbra{\frac{1}{\varepsilon^{\max\cbra{2/(q-1), 2}}}}$ in \cite{CW25}. 
In particular, for constant $q \geq 2$, the sample complexity is known to be $\widetilde{\Theta}\rbra{1/\varepsilon^2}$ \cite{CW25}. 
Specifically, a matching lower bound of $\Omega\rbra{\frac{1}{\varepsilon^2}}$ for $q = 2$ was given in \cite[Theorem 5]{CWLY23} and \cite[Lemma 3]{GHYZ24}, and a lower bound of $\Omega\rbra{\frac{1}{\varepsilon^{\max\cbra{1/(q-1), 2}}}}$ for constant $q > 1$ was given in \cite{CW25}. 
For every constant $0 < q < 1$, see \cref{sec:q-quantum-entropy} for further discussions, as we are only aware of a query complexity result in \cite{WGL+24}.

\subsubsection{Query complexity} \label{sec:q-quantum-entropy}

To further exploit the power of quantum computing, entropy estimation of both probability distributions and quantum states has also been considered in the quantum query input model, in the hope of further speedups. 

\paragraph{Purified quantum query access.}
The now standard quantum query model for probability distributions, called the \textit{purified quantum query access} model (see \cref{sec:def-quantum-query-model} for the formal definition), assumes a quantum unitary oracle $\mathcal{O}$ for the probability distribution $p$ such that
\begin{equation} \label{eq:def-O}
    \mathcal{O}\ket{0}_{\mathsf{A}}\ket{0}_{\mathsf{B}} = \sum_{j=0}^{N-1} \sqrt{p_j} \ket{j}_{\mathsf{A}} \ket{\phi_j}_{\mathsf{B}},
\end{equation}
where $\ket{\phi_j}$ is a (normalized) pure quantum state.\footnote{There are also other quantum models for probability distributions, which are less general than the purified quantum query access model.
For comparison and relationship of other models with the purified quantum query access model, the readers can refer to \cite{GL20,Bel19}.} 
This model was adopted in \cite{HM19,GL20,Bel19,GHS21,LWL24,WZL24,CRO24} for testing properties of probability distributions.
This is a natural quantum generalization, as one can draw a sample from the probability distribution $p$ by measuring the subsystem $\mathsf{A}$ of the quantum state in \cref{eq:def-O} prepared by $\mathcal{O}$ in the computational basis.
In addition, this quantum generalization is also reasonable, as one can implement such a quantum unitary oracle $\mathcal{O}$ for a probability distribution $p$ once the ``source code'' of the generator of $p$ is known. 
The query complexity of a quantum algorithm with access to the quantum unitary oracle $\mathcal{O}$ is measured by the number of queries to (controlled-)$\mathcal{O}$ and its inverse. For convenience, a query to $\mathcal{O}$ can refer to a query to $\mathcal{O}$, $\mathcal{O}^\dag$, controlled-$\mathcal{O}$, or controlled-$\mathcal{O}^\dag$. 

The purified quantum query access model can also be extended to testing properties of quantum states. 
For an $N$-dimensional quantum state $\rho$ with the spectrum decomposition
\begin{equation}
    \rho = \sum_{j=0}^{N-1} p_j \ketbra{\psi_j}{\psi_j},
\end{equation}
a quantum unitary oracle $\mathcal{O}$, also understood as a state-preparation circuit, for the quantum state $\rho$ prepares a purification of $\rho$, i.e.,
\begin{equation} \label{eq:def-O-rho}
    \mathcal{O}\ket{0}_{\mathsf{A}}\ket{0}_{\mathsf{B}} = \sum_{j=0}^{N-1} \sqrt{p_j} \ket{\psi_j}_{\mathsf{A}} \ket{\phi_j}_{\mathsf{B}} \coloneqq \ket{\rho}_{\mathsf{AB}},
\end{equation}
and a sample of $\rho$ can be obtained by tracing out the subsystem $\mathsf{B}$: $\rho = \tr_{\mathsf{B}}\rbra{\ketbra{\rho}{\rho}_{\mathsf{AB}}}$. 
In this way, a probability distribution $p$ can be viewed as a quantum state (by taking $\ket{\psi_j} = \ket{j}$):
\begin{equation} \label{eq:rho-p}
    \rho = \sum_{j=0}^{N-1} p_j \ketbra{j}{j},
\end{equation}
where $\cbra{\ket{j}}$ is the computational basis and measuring $\rho$ in the computational basis will give the outcome $j$ with probability $p_j$. 
This model was adopted in quantum computational complexity \cite{Wat02,BASTS10,RASW23,LGLW23} and in quantum algorithms \cite{GL20,GLM+22,WZC+23,GP22,WZ24b}. 

\paragraph{Shannon entropy and von Neumann entropy.}
Estimating the Shannon entropy with a quadratic quantum speedup in $N$ over classical algorithms was first noted in \cite{BHH11}. 
Later in \cite{LW19}, the quantum query complexity of Shannon entropy estimation was shown to be $\widetilde{O}\rbra{\frac{\sqrt{N}}{\varepsilon^2}}$, which was later improved to $\widetilde{O}\rbra{\frac{\sqrt{N}}{\varepsilon^{1.5}}}$ in \cite{GL20}, and an almost matching lower bound of $\widetilde{\Omega}\rbra{\sqrt{N}}$ was shown in
\cite{BKT20} for constant $\varepsilon$.
In addition, the quantum query complexity of Shannon entropy estimation with multiplicative error $\gamma$ was shown in \cite{GHS21} to be $\widetilde{O}\rbra{N^{\rbra{1+\eta}/2\gamma^2}}$ and $\Omega\rbra{N^{1/3\gamma^2}}$ provided that $\mathrm{H}\rbra{p} = \Omega\rbra{\gamma+\frac{1}{\eta}}$.

The quantum query complexity of von Neumann entropy estimation was shown to be $\widetilde{O}\rbra{\frac{N}{\varepsilon^{1.5}}}$ in \cite{GL20}.
When the quantum state is of rank $r$, the quantum query complexity can be improved to $\widetilde{O}\rbra{\frac{r}{\varepsilon^2}}$ in \cite{WGL+24}. 
In addition, the quantum query complexity of von Neumann entropy estimation with multiplicative error $\gamma$ was shown in \cite{GHS21} to be $\widetilde{O}\rbra{N^{1/2+\rbra{1+\eta}/2\gamma^2}}$.

\paragraph{R\'enyi entropy.}
Quantum query complexity of R\'enyi entropy estimation of probability distributions was systematically studied in \cite{LW19}, which was shown to be $\widetilde{O}\rbra{\frac{N^{1/q-1/2}}{\varepsilon^2}}$ and $\Omega\rbra{\frac{N^{1/7q-o\rbra{1}}}{\varepsilon^{2/7}} + \frac{N^{1/3}}{\varepsilon^{1/6}}}$ for every constant $0 < q < 1$,  $\widetilde{O}\rbra{\frac{N^{1-1/2q}}{\varepsilon^2}}$ and $\Omega\rbra{\frac{N^{1/3}}{\varepsilon^{1/6}}+\frac{N^{1/2-1/2q}}{\varepsilon}}$ for every constant non-integer $q > 1$, and $\widetilde{O}\rbra{\frac{N^{\nu\rbra{1-1/q}}}{\varepsilon^2}}$ for constant integer $q \geq 2$ where $\nu = 1 - \frac{2^{q-2}}{2^{q}-1} < \frac{3}{4}$.
These quantum query complexities were later improved in \cite{WZL24} to $\widetilde{O}\rbra{\frac{N^{1/2q}}{\varepsilon^{1+1/2q}}}$ for every constant $0 < q < 1$ and $\widetilde{O}\rbra{\frac{N^{1-1/2q}}{\varepsilon}+\frac{\sqrt{N}}{\varepsilon^{1+1/2q}}}$ for every constant $q > 1$. 

The quantum query complexity of R\'enyi entropy estimation of quantum states was initially considered in \cite{SH21}. 
In \cite{WGL+24}, it was shown that if the quantum state is of rank $r$, then the quantum query complexity will be $\poly\rbra{r, \frac{1}{\varepsilon}}$ for any constant non-integer $q > 0$.
Later in \cite{WZL24}, the quantum query complexity was shown to be $\widetilde{O}\rbra{\frac{N^{1/2+1/2q}}{\varepsilon^{1+1/2q}}}$ for every constant $0 < q < 1$ and $\widetilde{O}\rbra{\min\cbra{\frac{N}{\varepsilon^{1+1/2q}}+\frac{N^{3/2-1/2q}}{\varepsilon}, \frac{N}{\varepsilon^{1+1/q}}}}$ for every constant $q > 1$. 

\paragraph{Tsallis entropy estimation.}
For integer $q \geq 2$, the Tsallis entropy of order $q$ can be estimated with quantum query complexity $O\rbra{\frac{1}{\varepsilon}}$ via the SWAP test \cite{BCWdW01} (for $q = 2$) and the Shift test \cite{EAO+02} (for $q \geq 2$) (see explanations in \cref{sec:warmup}). 
For constant $0 < q < 1$, the quantum query complexity of Tsallis entropy estimation of quantum states was shown to be $\widetilde{O}\rbra{\frac{N^{\rbra{3-q^2}/2q}}{\varepsilon^{\rbra{3+q}/2q}}}$ in \cite{WGL+24}.
For constant non-integer $q > 1$, the quantum query complexity was shown to be $\poly\rbra{r, \frac{1}{\varepsilon}}$ if the quantum state is of rank $r$ in \cite{WGL+24}, which was later exponentially improved to $O\rbra{\frac{1}{\varepsilon^{1+{1}/\rbra{q-1}}}}$ in \cite{LW25}.
A lower bound of $\Omega\rbra{\frac{1}{\sqrt{\varepsilon}}}$ for general $q \geq 1 + \Omega\rbra{1}$ was given in \cite[Theorem 5.7]{LW25}.

\subsubsection{Under other models or constraints}

There are also several quantum entropy estimators that assume prior knowledge of (an upper bound on) the reciprocal of the minimum non-zero eigenvalue of quantum states, denoted by $\kappa$. 
A quantum estimator for von Neumann entropy was given in \cite{CLW20} with query complexity $\widetilde{O}\rbra{\frac{\kappa^2}{\varepsilon}}$. 
A quantum estimator for R\'enyi entropy of (non-integer) order $q$ was given in \cite{SH21} with query complexity $\widetilde{O}\rbra{\frac{\kappa N^{\max\cbra{q, 1}}}{\varepsilon^2}}$. 
A family of entropy estimators with sample access were studied in \cite{WZW23}, where they gave a von Neumann entropy estimator with sample complexity $\widetilde{O}\rbra{\frac{\kappa^2}{\varepsilon^5}}$. 

In addition, a classical algorithm for estimating the von Neumann entropy of density matrices was proposed in \cite{KDS+20}. 
In \cite{GH20}, they showed that estimating the (von Neumann) entropy of shallow quantum circuits is generally hard. 

\section{Preliminaries}

In this section, we include the necessary preliminaries for quantum computing. 
A pure quantum state is denoted by a vector $\ket{\psi}$ in a Hilbert space, with its norm denoted by $\Abs{\ket{\psi}} = \sqrt{\braket{\psi}{\psi}}$, where $\braket{\varphi}{\psi}$ is the inner product of $\ket{\varphi}$ and $\ket{\psi}$. 
The tensor product of two pure quantum states $\ket{\varphi}$ and $\ket{\psi}$ is denoted by $\ket{\varphi} \otimes \ket{\psi}$ or $\ket{\varphi}\ket{\psi}$. 
The operator norm of a linear operator $A$ is defined by 
\begin{equation}
    \Abs{A} = \sup_{\Abs{\ket{x}} = 1} \Abs{A\ket{x}}.
\end{equation}
A linear operator $U$ is unitary if $U^\dag U = UU^\dag = I$, where $I$ is the identity operator and $U^\dag$ is the Hermitian conjugate of $U$. 

\subsection{Quantum query models} \label{sec:def-quantum-query-model}

The quantum query model used in this paper is usually known as the purified quantum query access model. 
In the following, we define the state-preparation oracle for mixed quantum states. 

\begin{definition} [State-preparation oracles] \label{def:state-preparation}
   An $\rbra{a+b}$ unitary operator $U$ is said to prepare an $a$-qubit mixed quantum state $\rho$, if $U\ket{0}_{\mathsf{A}}\ket{0}_{\mathsf{B}} = \ket{\psi}_{\mathsf{AB}}$ and $\rho = \tr_{\mathsf{B}}\rbra{\ketbra{\psi}{\psi}_{\mathsf{AB}}}$, where the subscripts $\mathsf{A}$ and $\mathsf{B}$ respectively denote the subsystem of the first $a$ qubits and the last $b$ qubits.
\end{definition}

The quantum query oracle for probability distributions is a special case of the state-preparation oracle for mixed quantum states. 

\begin{definition} [Quantum query oracle for probability distributions] \label{def:oracle-prob}
    A unitary operator is said to be a quantum query oracle for an $N$-dimensional probability distribution $p = \rbra{p_0, p_1, \dots, p_{N-1}} \in \mathbb{R}^N$, if it prepares the $\ceil{\log_2\rbra{N}}$-qubit mixed quantum state
    \begin{equation}
        \sum_{j=0}^{N-1} p_j\ketbra{j}{j}.
    \end{equation}
\end{definition}

\subsection{Quantum singular value transformation}

Quantum singular value transformation (QSVT) \cite{GSLW19} is a useful tool in the design of quantum algorithms. 
To formally state this tool, we introduce the notion of block-encoding. 

\begin{definition} [Block-encoding] \label{def:block-encoding}
    Suppose that $A$ is an $n$-qubit operator, $\alpha, \varepsilon \geq 0$, and $a \in \mathbb{N}$. 
    Then, an $\rbra{n+a}$-qubit unitary operator $U$ is said to be a $\rbra{\alpha, a, \varepsilon}$-block-encoding of $A$, if 
    \begin{equation}
        \Abs*{ \alpha \rbra*{\bra{0}^{\otimes a} \otimes I_n} U \rbra*{\ket{0}^{\otimes a} \otimes I_n} - A } \leq \varepsilon,
    \end{equation}
    where $I_n$ is the identity operator on $n$ qubits. 
\end{definition}

In this paper, we need a special case of QSVT, stated as follows. 

\begin{theorem} [Quantum singular value transformation for even/odd polynomials, adapted from {\cite[Theorem 2]{GSLW19}}] \label{thm:qsvt}
    Let $U$ be a unitary operator that is an $\rbra{1, a, 0}$-block-encoding of an Hermitian operator $A$, where $a \geq 1$ is an integer. 
    Suppose that $p \in \mathbb{R}\sbra{x}$ is an even/odd polynomial of degree $d$ such that $\abs{p\rbra{x}} \leq 1$ for all $x \in \sbra{-1, 1}$.
    Then, there is a quantum circuit $W_0$ that is a $\rbra{1, a+1, 0}$-block-encoding of $p\rbra{A}$, consisting of $O\rbra{d}$ queries to $U$. 
    Moreover, if $p$ is efficiently computable,\footnote{A polynomial $p\rbra{x} = c_0 + c_1 x + \dots + c_d x^d \in \mathbb{R}\sbra{x}$ of degree $d$ is efficiently computable, if for each $0 \leq j \leq d$, the coefficient $c_j$ can be computed to within additive error $\varepsilon$ in deterministic time $\poly\rbra{1/\varepsilon}$.
    In \cite{GSLW19}, they implicitly assumed that the polynomial $p$ is efficiently computable, and thus they only consider the case that $\delta$ is non-zero in order to ensure that the description of the quantum circuit $W_{\delta}$ can be efficiently computed.
    Their approach also leads to the case of $\delta = 0$, regardless of the computability.}
    then for every $\delta \in \rbra{0, 1}$, there is a quantum circuit $W_{\delta}$ that is a $\rbra{1, a+1, \delta}$-block-encoding of $p\rbra{A}$, consisting of $O\rbra{d}$ queries to $U$ and $O\rbra{ad}$ two-qubit gates, and the description of $W_{\delta}$ can be computed on a classical computer in $\poly\rbra{d, \log\rbra{\frac{1}{\delta}}}$ time. 
\end{theorem}

\subsection{Polynomial approximation}

To apply the QSVT, we need the polynomial approximation for monomials given in \cite{SV14}. 

\begin{theorem} [Polynomial approximation for monomials, {\cite[Theorem 3.3]{SV14}}] \label{thm:poly-approx-monomial}
    For any integer $q \geq 1$ and real number $\varepsilon \in \rbra{0, 1}$, there is an efficiently computable even/odd polynomial $p \in \mathbb{R}\sbra{x}$ of degree $O\rbra{\sqrt{q\log\rbra{\frac{1}{\varepsilon}}}}$ and parity $q \bmod 2$ such that $\abs{p\rbra{x} - x^q} \leq \varepsilon$ and $\abs{p\rbra{x}} \leq 1$ for all $x \in \sbra{-1, 1}$.
\end{theorem}

\subsection{Quantum algorithmic tools}

When given the state-preparation circuit of a mixed quantum state $\rho$, it is usually needed to implement a unitary block-encoding of $\rho$. 
In the following, we include the subroutine for this. 

\begin{lemma} [Block-encoding of density operators, {\cite[Lemma 7]{LC19}} and {\cite[Lemma 25]{GSLW19}}] \label{lemma:block-encoding-density}
    Suppose that an $\rbra{n+a}$-qubit unitary operator $U$ prepares an $n$-qubit mixed quantum state $\rho$. 
    Then, we can implement a $\rbra{2n+a}$-qubit quantum circuit $W$ that is a $\rbra{1, n+a, 0}$-block-encoding of $\rho$, using $2$ queries to $U$ and $O\rbra{n}$ two-qubit gates.
\end{lemma}

Once we have a unitary operator that is a block-encoding of a matrix $A$, we can estimate $\tr\rbra{A\rho}$ by the Hadamard test \cite{AJL09} on the quantum state $\rho$. 
Here, we use the version in \cite{GP22}. 

\begin{lemma} [Hadamard test for block-encodings, {\cite[Lemma 9]{GP22}}] \label{lemma:hadamard-test}
    Suppose that an $\rbra{n+a}$-qubit unitary operator $U$ is a $\rbra{1, a, 0}$-block-encoding of $A$. 
    Then, we can implement an $\rbra{n+a+1}$-qubit quantum circuit $W$ that on input any $n$-qubit mixed quantum state $\rho$ outputs $0$ with probability $\frac{1 + \Real\rbra{\tr\rbra{A\rho}}}{2}$, using $1$ query to $U$ and $O\rbra{1}$ two-qubit gates. 
    In other words, the quantum circuit $W$ satisfies that
    \begin{equation}
        \tr\rbra*{\Pi W \rbra*{\rho \otimes \ketbra{0}{0}^{\otimes \rbra{a+1}}} W^\dag } = \frac{1 + \Real\rbra{\tr\rbra{A\rho}}}{2},
    \end{equation}
    where $\Pi = \ketbra{0}{0} \otimes I$ is the projector onto the subspace with the first qubit being $\ket{0}$.
\end{lemma}

To efficiently estimate the amplitudes of a (pure) quantum state, we need the subroutine for quantum amplitude estimation. 

\begin{theorem} [Quantum amplitude estimation, {\cite[Theorem 12]{BHMT02}}] \label{thm:ampl-est}
    Suppose that $U$ is a unitary operator acting on a one-qubit subsystem $\mathsf{A}$ and an $n$-qubit sybsystem $\mathsf{B}$ such that 
    \begin{equation}
        U\ket{0}_{\mathsf{A}}\ket{0}_{\mathsf{B}} = \sqrt{p} \ket{0}_\mathsf{A} \ket{\phi_0}_{\mathsf{B}} + \sqrt{1 - p} \ket{1}_{\mathsf{A}} \ket{\phi_1}_{\mathsf{B}},
    \end{equation}
    where $p \in \sbra{0, 1}$ and $\ket{\phi_0}$ and $\ket{\phi_1}$ are normalized pure quantum states. 
    Then, for any $\varepsilon \in \rbra{0, 1}$, there is a quantum query algorithm that outputs an estimate $\tilde p \in \sbra{0, 1}$ of $p$ such that 
    \begin{equation}
        \Pr\sbra*{ \abs*{\tilde p - p} \leq \varepsilon } \geq \frac{2}{3},
    \end{equation}
    using $O\rbra{\frac{1}{\varepsilon}}$ queries to $U$ and $O\rbra{\frac{n}{\varepsilon}}$ two-qubit gates. 
\end{theorem}

\subsection{Quantum query lower bound for distinguishing probability distributions}

To show the quantum query lower bounds in this paper, we need the following result about the quantum query lower bound for distinguishing probability distributions. 

\begin{theorem} [Quantum query lower bound for distinguishing probability distributions, {\cite[Theorem 4]{Bel19}}] \label{thm:qlower-dis-prob}
    Given a quantum query oracle $\mathcal{O}$ for a probability distribution $P$ (defined by \cref{def:oracle-prob}), promised that $P$ is one of the two probability distributions $p$ and $q$ on a sample space of $N$ elements, any quantum query algorithm that determines whether $P = p$ or $P = q$ with probability $\geq \frac{2}{3}$ requires query complexity $\Omega\rbra{\frac{1}{d_\mathrm{H}\rbra{p, q}}}$, where 
    \begin{equation}
        d_\mathrm{H}\rbra{p, q} = \sqrt{\frac{1}{2} \sum_{j=0}^{N-1} \rbra*{\sqrt{p_j} - \sqrt{q_j}}^2}
    \end{equation}
    is the Hellinger distance. 
\end{theorem}

\section{Meta-Algorithm} \label{sec:mata-algo}

In this section, we present a meta-algorithm for Tsallis entropy estimation, which not only allows us to provide a quantum estimator for Tsallis entropy of integer order (in \cref{sec:upper}) but also enables us to prove the lower bound on the approximate degree of monomials $x^q$ (in \cref{sec:approx-deg}). 

\subsection{Description of the algorithm}

We first formalize the setting and the goal. 
\begin{itemize}
    \item \textbf{Setting}: We assume quantum query access to the state-preparation circuit of a mixed quantum state $\rho$ (defined by \cref{def:state-preparation}).
    To be precise, we assume that $\rho$ is an $n$-qubit mixed quantum state and the state-preparation circuit $\mathcal{O}$ is an $\rbra{n+a}$-qubit unitary operator. 
    \item \textbf{Goal}: Estimate the value of the $q$-Tsallis entropy of $\rho$, $\mathrm{S}_q\rbra{\rho}$, where $q \geq 2$ is an integer. 
\end{itemize}

Now we describe the meta-algorithm in the following six steps. 
See \cref{algo:tsallis} for its formal description. 

\begin{algorithm}[!htp]
    \caption{Meta-algorithm for quantum $q$-Tsallis entropy estimation.}
    \label{algo:tsallis}
    \begin{algorithmic}[1]
    \Require An $\rbra{n+a}$-qubit quantum unitary oracle $\mathcal{O}$ that prepares an $n$-qubit mixed quantum state $\rho$; parameters $\varepsilon_{\poly}, \varepsilon_{\textup{QSVT}}, \varepsilon_{\textup{QAE}} \in \sbra{0, 1}$ whose subscripts respectively stand for the polynomial approximation, QSVT, and quantum amplitude estimation. 

    \Ensure An estimate $\tilde S$ of $\mathrm{S}_q\rbra{\rho}$.
    
    \State Let $U_\rho$ be the $\rbra{2n+a}$-qubit unitary operator (by \cref{lemma:block-encoding-density}) that is a $\rbra{1, n+a, 0}$-block-encoding of $\rho$, using $O\rbra{1}$ queries to $\mathcal{O}$.

    \State Let $p \in \mathbb{R}\sbra{x}$ be an even/odd polynomial with parity $\rbra{q-1} \bmod 2$ such that $\abs{p\rbra{x} - x^{q-1}} \leq \varepsilon_{\poly}$ and $\abs{p\rbra{x}} \leq 1$ for all $x \in \sbra{-1, 1}$.

    \State Let $U_{p\rbra{\rho}}$ be the $\rbra{2n+a+1}$-qubit unitary operator (by \cref{thm:qsvt}) that is a $\rbra{1, n+a+1, \varepsilon_{\textup{QSVT}}}$-block-encoding of $p\rbra{\rho}$, using $O\rbra{\deg\rbra{p}}$ queries to $U_\rho$. 

    \State Let $V$ be the $\rbra{2n+2a+2}$-qubit unitary operator for the Hadamard test (by \cref{lemma:hadamard-test}) such that the measurement outcome of its first qubit is $0$ with probability $\gamma \approx \frac{1+\tr\rbra{\rho p\rbra{\rho}}}{2}$, using $1$ query to each of $U_{p\rbra{\rho}}$ and $\mathcal{O}$. 

    \State Let $\tilde \gamma$ be the estimate of $\gamma$ (by \cref{thm:ampl-est}) such that $\abs{\tilde \gamma - \gamma} \leq \varepsilon_{\textup{QAE}}$ with probability $\geq \frac{2}{3}$, using $O\rbra{\frac{1}{\varepsilon_{\textup{QAE}}}}$ queries to $V$.

    \State \Return $\frac{2\rbra{1-\tilde\gamma}}{q-1}$.
    
    \end{algorithmic}
\end{algorithm}

\textbf{Step 1: Implement a unitary block-encoding of $\rho$.} 
This can be simply done by \cref{lemma:block-encoding-density}, which allows us to implement a $\rbra{2n+a}$-qubit quantum circuit $U_\rho$ that is a $\rbra{1, n+a, 0}$-block-encoding of $\rho$, using $O\rbra{1}$ queries to $\mathcal{O}$ and $O\rbra{n}$ additional two-qubit gates. 

\textbf{Step 2: Find a polynomial approximating $x^{q-1}$.}
The required polynomial $p \in \mathbb{R}\sbra{x}$ should satisfy the three conditions: 
\begin{enumerate}[(i)]
    \item $\abs{p\rbra{x} - x^{q-1}} \leq \varepsilon_{\poly}$ for all $x \in \sbra{-1, 1}$, where $\varepsilon_{\poly} \in \sbra{0, 1}$ is a freely specified parameter. \label{item:cond1-poly}
    \item $\abs{p\rbra{x}} \leq 1$ for all $x \in \sbra{-1, 1}$. \label{item:cond2-poly}
    \item $p$ is either even or odd with parity $\rbra{q-1} \bmod 2$. \label{item:cond3:poly}
\end{enumerate}
Condition \ref{item:cond1-poly} means that the polynomial $p$ is uniformly $\varepsilon_{\poly}$-close to the monomial $x^{q-1}$. 
Conditions \ref{item:cond2-poly} and \ref{item:cond3:poly} are for the purpose of performing QSVT for the even/odd polynomial $p$ later (as required in \cref{thm:qsvt}). 
The existence of polynomials that satisfy the three conditions is guaranteed by \cref{thm:poly-approx-monomial}.

\textbf{Step 3: Implement a unitary block-encoding of $p\rbra{\rho}$.}
This is done by directly performing QSVT according to \cref{thm:qsvt}. 
Specifically, we can implement a $\rbra{2n+a+1}$-qubit quantum circuit $U_{p\rbra{\rho}}$ that is a $\rbra{1, n+a+1, \varepsilon_{\textup{QSVT}}}$-block-encoding of $p\rbra{\rho}$, using $O\rbra{\deg\rbra{p}}$ queries to $U_\rho$ and $O\rbra{\rbra{n+a}\deg\rbra{p}}$ additional two-qubit gates, where $\varepsilon_{\textup{QSVT}} \in \sbra{0, 1}$ is a parameter that can be specified freely. 
In particular, when $\varepsilon_{\textup{QSVT}} > 0$ and the polynomial $p$ is efficiently computable, the circuit description (in terms of gates and oracle queries) can be efficiently computed on a classical computer in $\poly\rbra{\deg\rbra{p}, \log\rbra{\frac{1}{\varepsilon_{\textup{QSVT}}}}}$ time. 

\textbf{Step 4: Implement an Hadamard test that outputs $0$ with probability $\approx \frac{1+\tr\rbra{\rho p\rbra{\rho}}}{2}$.}
This is done by applying the Hadamard test specified in \cref{lemma:hadamard-test} that uses $1$ query to the unitary operator $U_{p\rbra{\rho}}$ and $1$ sample of the mixed quantum state $\rho$ (preparable by $1$ query to $\mathcal{O}$). 
Let $V$ be the $\rbra{2n+2a+2}$-qubit quantum circuit of this Hadamard test, using $1$ query to each of $U_{p\rbra{\rho}}$ and $\mathcal{O}$ and $O\rbra{1}$ additional two-qubit gates. 
If $U_{p\rbra{\rho}}$ is a $\rbra{1, n+a+1, 0}$-block-encoding of $A$, then measuring the first qubit of $V\ket{0}^{\otimes \rbra{2n+2a+2}}$ gives $0$ with probability 
\begin{equation} \label{eq:def-gamma}
    \gamma = \frac{1+\Real\rbra{\tr\rbra{A\rho}}}{2} \in \sbra{0, 1}.
\end{equation}
In other words, $V$ is of the form
\begin{equation}
    V\ket{0}^{\otimes \rbra{2n+2a+2}} = \sqrt{\gamma} \ket{0} \ket{\phi_0} + \sqrt{1 - \gamma} \ket{1} \ket{\phi_1},
\end{equation}
where $\ket{\phi_0}$ and $\ket{\phi_1}$ are two (normalized) $\rbra{2n+2a+1}$-qubit pure quantum states. 

\textbf{Step 5: Estimate $\gamma$.}
This can be done by quantum amplitude estimation (\cref{thm:ampl-est}). 
For any $\varepsilon_{\textup{QAE}} \in \rbra{0, 1}$, we can obtain an estimate $\tilde \gamma$ of $\gamma$ such that
\begin{equation} \label{eq:diff-gamma}
    \Pr \sbra*{ \abs*{\tilde \gamma - \gamma} \leq \varepsilon_{\textup{QAE}} } \geq \frac{2}{3},
\end{equation}
using $O\rbra{\frac{1}{\varepsilon_{\textup{QAE}}}}$ queries to $V$ and $O\rbra{\frac{n+a}{\varepsilon_{\textup{QAE}}}}$ additional two-qubit gates. 
Here, note that the parameter $\varepsilon_{\textup{QAE}}$ can be specified freely. 

\textbf{Step 6: Return $\frac{2\rbra{1-\tilde\gamma}}{q-1}$ as an estimate of $\mathrm{S}_q\rbra{\rho}$.}

\subsection{Correctness}

We summarize the properties of \cref{algo:tsallis} in the following lemma. 

\begin{lemma} \label{lemma:meta}
    By specifying the three parameters $\varepsilon_{\poly}, \varepsilon_{\textup{QSVT}}, \varepsilon_{\textup{QAE}}$, with probability $\geq \frac{2}{3}$, the estimate of $\mathrm{S}_q\rbra{\rho}$ returned by
    \cref{algo:tsallis} is bounded by
    \begin{equation}
        \abs*{\frac{2\rbra{1-\tilde \gamma}}{q-1} - \mathrm{S}_q\rbra{\rho}} \leq \frac{1}{q-1} \rbra*{\varepsilon_{\textup{QSVT}} + \varepsilon_{\poly} + 2\varepsilon_{\textup{QAE}}}.
    \end{equation}
    In addition, \cref{algo:tsallis} uses $O\rbra{\frac{\deg\rbra{p}}{\varepsilon_{\textup{QAE}}}}$ queries to $\mathcal{O}$ and $O\rbra{\frac{\rbra{n+a}\deg\rbra{p}}{\varepsilon_{\textup{QAE}}}}$ two-qubit gates. 
    Moreover, if $\varepsilon_{\textup{QSVT}} > 0$ and $p$ is efficiently computable, the quantum circuit description of \cref{algo:tsallis} can be computed on a classical computer in $\poly\rbra{\deg\rbra{p}, n, a, \frac{1}{\varepsilon_{\textup{QAE}}}, \log\rbra{\frac{1}{\varepsilon_{\textup{QSVT}}}}}$ time. 
\end{lemma}

\begin{proof}
    To analyze the error, we first show the following inequality. 
    \begin{equation} \label{eq:gamma-vs-tr}
        \abs*{\rbra{2 \gamma - 1} - \tr\rbra{\rho^q}} \leq \varepsilon_{\textup{QSVT}} + \varepsilon_{\poly}.
    \end{equation}
    To see this, by \cref{eq:def-gamma}, \cref{eq:gamma-vs-tr} simplifies to 
    \begin{equation} \label{eq:diff-Re-tr-pow}
        \abs*{ \Real\rbra{\tr\rbra{A\rho}} -  \tr\rbra{\rho^q} } \leq \varepsilon_{\textup{QSVT}} + \varepsilon_{\poly}.
    \end{equation}
    Given the fact that $\abs{\Real\rbra{x} - \Real\rbra{y}} \leq \abs{x - y}$ for any $x, y \in \mathbb{C}$ and that $\tr\rbra{\rho^q} \in \mathbb{R}$, we have
    \begin{align}
        \abs*{\Real\rbra{\tr\rbra{A\rho}} -  \tr\rbra{\rho^q}} 
        & \leq \abs*{ \tr\rbra{A\rho} -  \tr\rbra{\rho^q}} \\
        & \leq \Abs*{A - \rho^{q-1}}, \label{eq:diff-A-rho-pow}
    \end{align}
    where the last inequality uses the fact that $\abs{\tr\rbra{AB}} \leq \Abs{A} \tr\rbra{\abs{B}}$ for any two square matrices $A$ and $B$ (which is a special case of the matrix H\"older inequality, e.g., \cite[Theorem 2]{Bau11}). 
    As $U_{p\rbra{\rho}}$ is a $\rbra{1, n+a+1, 0}$-block-encoding of $A$ and also a $\rbra{1, n+a+1, \varepsilon_{\textup{QSVT}}}$-block-encoding of $p\rbra{\rho}$, by the definition of block-encoding (\cref{def:block-encoding}), we have
    \begin{equation} \label{eq:diff-A-p-rho}
        \Abs{A - p\rbra{\rho}} \leq \varepsilon_{\textup{QSVT}}.
    \end{equation}
    By condition \ref{item:cond1-poly} in Step 2 for the polynomial $p$, we have
    \begin{equation} \label{eq:diff-p-rho-rho-pow}
        \Abs*{p\rbra{\rho} - \rho^{q-1}} \leq \varepsilon_{\poly
        },
    \end{equation}
    since the eigenvalues of $\rho$ lie in the range $\sbra{0, 1}$. 
    Combining \cref{eq:diff-A-p-rho,eq:diff-p-rho-rho-pow}, we have
    \begin{align}
        \eqref{eq:diff-A-rho-pow}
        & \leq \Abs{A - p\rbra{\rho}} + \Abs*{p\rbra{\rho} - \rho^{q-1}} \\
        & \leq \varepsilon_{\textup{QSVT}} + \varepsilon_{\poly},
    \end{align}
    thereby yielding \cref{eq:diff-Re-tr-pow} and thus \cref{eq:gamma-vs-tr}. 
    By \cref{eq:diff-gamma,eq:gamma-vs-tr}, we have 
    \begin{equation}
        \Pr \sbra*{ \abs*{ \rbra{2\tilde \gamma - 1} - \tr\rbra{\rho^q} } \leq \varepsilon_{\textup{QSVT}} + \varepsilon_{\poly} + 2\varepsilon_{\textup{QAE}} } \geq \frac{2}{3}.
    \end{equation}
    Noting that $\mathrm{S}_q\rbra{\rho} = \frac{1-\tr\rbra{\rho^q}}{q-1}$, we finally yield 
    \begin{equation}
        \Pr \sbra*{ \abs*{ \frac{2\rbra{1-\tilde\gamma}}{q-1} - \mathrm{S}_q\rbra{\rho} } \leq \frac{1}{q-1} \rbra*{\varepsilon_{\textup{QSVT}} + \varepsilon_{\poly} + 2\varepsilon_{\textup{QAE}}} } \geq \frac{2}{3}.
    \end{equation}

    To complete the proof, now we analyze the complexity of \cref{algo:tsallis}. 
    We list the query complexity of each unitary operator in \cref{algo:tsallis}: $U_\rho$ uses $O\rbra{1}$ queries to $\mathcal{O}$, $U_{p\rbra{\rho}}$ uses $O\rbra{\deg\rbra{p}}$ queries to $\mathcal{O}$, $V$ uses $O\rbra{\deg\rbra{p}}$ queries to $\mathcal{O}$, and finally the overall algorithm uses $O\rbra{\frac{1}{\varepsilon_{\textup{QAE}}}}$ queries to $V$. 
    From this, we can see that \cref{algo:tsallis} uses $O\rbra{\frac{\deg\rbra{p}}{\varepsilon_{\textup{QAE}}}}$ queries to $\mathcal{O}$. 
    The number of two-qubit gates used in \cref{algo:tsallis} can be analyzed similarly. 
\end{proof}

\section{Quantum Query Complexity of Tsallis Entropy Estimation}

In this section, we first provide a quantum query algorithm for estimating the Tsallis entropy of integer order in \cref{sec:upper} based on the meta-algorithm given in \cref{sec:mata-algo}, and then establish a matching lower bound (up to polylogarithmic factors) in \cref{sec:lower}. 

\subsection{Upper bounds} \label{sec:upper}

\begin{theorem} [Optimal quantum estimator for Tsallis entropy] \label{thm:q-tsallis-estimator}
    Given quantum query access to the state-prepraration circuit of an $n$-qubit mixed quantum state $\rho$, for every integer $q \geq 2$ and $\varepsilon \in (0, 1/q]$, there is a quantum query algorithm that estimates $\mathrm{S}_q\rbra{\rho}$ to within additive error $\varepsilon$ with query complexity $O\rbra{\frac{\sqrt{\log\rbra{1/q\varepsilon}}}{\sqrt{q}\varepsilon}}$. 
    Moreover, if the state-preparation circuit acts on $m$ qubits, then the quantum query algorithm uses $O\rbra{\frac{m\sqrt{\log\rbra{1/q\varepsilon}}}{\sqrt{q}\varepsilon}}$ additional two-qubit gates, and the circuit description of the quantum algorithm can be computed on a classical computer in $\poly\rbra{q, m, \frac{1}{\varepsilon}}$ time. 
\end{theorem}

\begin{proof}
    This can be achieved by applying \cref{lemma:meta} with the parameters $\varepsilon_{\textup{QSVT}} = \varepsilon_{\poly} = \varepsilon_{\textup{QAE}} = \frac{\rbra{q-1}\varepsilon}{4} = \Theta\rbra{q\varepsilon}$, which gives the quantum query complexity $O\rbra{\frac{\deg\rbra{p}}{\varepsilon_{\textup{QAE}}}} = O\rbra{\frac{\deg\rbra{p}}{q\varepsilon}}$, provided that there exists an efficiently computable polynomial $p \in \mathbb{R}\sbra{x}$ satisfying Conditions \ref{item:cond1-poly}, \ref{item:cond2-poly}, and \ref{item:cond3:poly}. 
    To complete the proof, we choose the polynomial specified in \cref{thm:poly-approx-monomial} with parameters $\varepsilon \coloneqq \varepsilon_{\poly}$, which is of degree $\deg\rbra{p} = O\rbra{\sqrt{q\log\rbra{\frac{1}{\varepsilon_{\poly}}}}} = O\rbra{\sqrt{q\log\rbra{\frac{1}{q\varepsilon}}}}$.
    Therefore, the overall quantum query complexity is 
    \begin{equation}
        O\rbra*{\frac{\deg\rbra{p}}{\varepsilon_{\textup{QAE}}}} = O\rbra*{\frac{\sqrt{q\log\rbra{1/q\varepsilon}}}{q\varepsilon}} = O\rbra*{\frac{\sqrt{\log\rbra{1/q\varepsilon}}}{\sqrt{q} \varepsilon}}.
    \end{equation}
    The number of two-qubit gates used in the algorithm can be analyzed similarly. 
\end{proof}

\subsection{Lower bounds} \label{sec:lower}

Our quantum query lower bound for Tsallis entropy estimation is given as follows. 

\begin{theorem} [Quantum lower bound for every integer $q \geq 2$] \label{thm:lb-summary}
    Given quantum query access to the oracle for a probability distribution $p$, for every integer $q \geq 2$ and sufficiently small $\varepsilon > 0$, any quantum query algorithm that estimates $\mathrm{H}_q\rbra{p}$ to within additive error $\varepsilon$ requires query complexity $\Omega\rbra{\frac{1}{\sqrt{q}\varepsilon}}$.
\end{theorem}

\begin{proof}
    We provide a lower bound of $\Omega\rbra{\frac{1}{\sqrt{q}\varepsilon}}$ for sufficiently large $q \geq 3$ in \cref{thm:lb} and a lower bound of $\Omega\rbra{\frac{1}{\varepsilon}}$ for constant $q \geq 2$ in \cref{thm:lb-constant-q}. 
    The two theorems together yield the proof. 
\end{proof}

We first prove the case of sufficiently large $q \geq 3$.

\begin{theorem} [Quantum lower bounds for sufficiently large $q \geq 3$] \label{thm:lb}
    Given quantum query access to the oracle for a probability distribution $p$, for every constant $t \in \rbra{0, \frac 1 e}$ and sufficiently integer $q \geq 3$, any quantum query algorithm that estimates $\mathrm{H}_q\rbra{p}$ to within additive error $\varepsilon$ requires query complexity $\Omega\rbra{\frac{1}{\sqrt{q}\varepsilon}}$ for all $\varepsilon \in \rbra{0, \frac{t}{q}}$.
\end{theorem}

\begin{proof}
    We start with considering the problem of distinguishing the two probability distributions $p^\pm$ on the sample space $\cbra{0, 1}$, defined by
    \begin{align}
        p^\pm_0 & = 1 - \frac{1}{q} \pm \delta, \\
        p^\pm_1 & = \frac{1}{q} \mp \delta,
    \end{align}
    where $\delta \in (0, \frac{1}{q}]$ is to be determined later. 
    This hard instance was initially used in \cite{CWYZ25} where they considered the sample complexity of distinguishing two quantum states with their spectra being the probability distributions $p^{\pm}$. 

    The Hellinger distance between the two distributions $p^+$ and $p^-$ is 
    \begin{align}
        d_{\mathrm{H}}\rbra{p^+, p^-} 
        & = \sqrt{1 - \sqrt{\rbra*{1 - \frac{1}{q}}^2 - \delta^2} - \sqrt{\rbra*{\frac{1}{q}}^2 - \delta^2}} \\
        & \leq \sqrt{1 - \rbra*{\rbra*{1 - \frac{1}{q}} - \frac{\delta^2}{1-\frac{1}{q}}} - \rbra*{\frac{1}{q} - {q\delta^2}}} \label{eq:dHgeq} \\
        & = \frac{q\delta}{\sqrt{q-1}} \\
        & \leq O\rbra{\sqrt{q} \delta}, \label{eq:dH-p-pm}
    \end{align}
    where \cref{eq:dHgeq} uses the fact that $\sqrt{a - x} \geq \sqrt{a} - \frac{x}{\sqrt{a}}$ for all $0 \leq x \leq a$.

    Now suppose that we are given a quantum query oracle $\mathcal{O}_{p}$ for an unknown probability distribution $p$ on the sample space $\cbra{0, 1}$, i.e., $\mathcal{O}_{p}$ prepares the mixed quantum state 
    \begin{equation}
        \rho_p = p_0\ketbra{0}{0} + p_1\ketbra{1}{1},
    \end{equation}
    where it is promised that $p$ is either $p^+$ or $p^-$. 
    Then, we consider the difference between $\mathrm{H}_q\rbra{p^+}$ and $\mathrm{H}_q\rbra{p^-}$ as follows.
    \begin{align}
        \mathrm{H}_q\rbra{p^-} - \mathrm{H}_q\rbra{p^+} 
        & = \frac{1}{q-1}\rbra*{\rbra*{1-\frac{1}{q}+\delta}^q - \rbra*{1-\frac{1}{q}-\delta}^q + \rbra*{\frac{1}{q}-\delta}^q - \rbra*{\frac{1}{q}+\delta}^q} \\
        & = \frac{2}{q-1} \sum_{j=0}^{\floor{\frac{q-1}{2}}} \binom{q}{2j+1} \rbra*{\rbra*{1-\frac{1}{q}}^{q-2j-1} - \rbra*{\frac{1}{q}}^{q-2j-1}} \delta^{2j+1} \label{eq:Sq-diff} \\
        & > \frac{2}{q-1} \binom{q}{1} \rbra*{\rbra*{1-\frac{1}{q}}^{q-1} - \rbra*{\frac{1}{q}}^{q-1}} \delta^1 \label{eq:Sq-diff-ge} \\
        & = 2\rbra*{\rbra*{1-\frac{1}{q}}^q - \rbra*{1-\frac{1}{q}}\rbra*{\frac{1}{q}}^{q-1}} \delta. \label{eq:Sq-diff-q-delta}
    \end{align}
    where \cref{eq:Sq-diff} uses the fact that
    \begin{align}
        \rbra{a+x}^k - \rbra{a-x}^k
        & = \sum_{j=0}^k \binom{k}{j} a^{k-j} x^j - \sum_{j=0}^k \binom{k}{j} a^{k-j} \rbra{-x}^j \\
        & = 2\sum_{j = 0}^{\floor{\frac{k-1}{2}}} \binom{k}{2j+1} a^{k-2j-1} x^{2j+1}
    \end{align}
    for any real numbers $a \geq 0$ and $x \geq 0$ and integer $k \geq 0$, and \cref{eq:Sq-diff-ge} assumes that $q > 2$.
    Note that
    \begin{equation}
        \lim_{q \to \infty} \rbra*{\rbra*{1-\frac{1}{q}}^q - \rbra*{1-\frac{1}{q}}\rbra*{\frac{1}{q}}^{q-1}} = \frac{1}{e}.
    \end{equation}
    Therefore, for every constant $t \in \rbra{0, \frac{1}{e}}$, we have by \cref{eq:Sq-diff-q-delta} that 
    \begin{equation}
        \mathrm{H}_q\rbra{p^-} - \mathrm{H}_q\rbra{p^+} > 2t\delta
    \end{equation}
    for sufficiently large $q$.
    
    By taking $\varepsilon = t\delta \in \rbra{0, \frac{t}{q}}$, we have $\mathrm{H}_q\rbra{p^-} - \mathrm{H}_q\rbra{p^+} > 2\varepsilon$, and thus every quantum query algorithm with query complexity $Q$ that estimates $\mathrm{H}_q\rbra{p}$ to within additive error $\varepsilon$ can be used to determine whether $p$ is $p^+$ or $p^-$ using $Q$ queries to $\mathcal{O}_p$.
    Specifically, we first obtain an estimate, $\tilde H$,  of $\mathrm{H}_q\rbra{p}$ to within additive error $\varepsilon$ with probability $\geq \frac{2}{3}$ via this algorithm using $Q$ queries to $\mathcal{O}_p$. 
    Then, answer $p = p^+$ if $\tilde H < \mathrm{H}_q\rbra{p^+} + \varepsilon$, and $p = p^-$ otherwise. 
    It can be seen that this approach correctly determines whether $p = p^+$ or $p = p^-$ with probability $\geq \frac{2}{3}$.
    On the other hand, by \cref{thm:qlower-dis-prob}, any quantum query algorithm that distinguishes $p^+$ and $p^-$ requires query complexity $\Omega\rbra{\frac{1}{d_{\mathrm{H}}\rbra{p^+, p^-}}}$. 
    Therefore, by \cref{eq:dH-p-pm}, we have
    \begin{equation}
        Q \geq \Omega\rbra*{\frac{1}{d_{\mathrm{H}}\rbra{p^+, p^-}}} = \Omega\rbra*{\frac{1}{\sqrt{q}\delta}} = \Omega\rbra*{\frac{t}{\sqrt{q}\varepsilon}} = \Omega\rbra*{\frac{1}{\sqrt{q}\varepsilon}},
    \end{equation}
    where the last equality uses that $t$ is constant. 
\end{proof}

Then, we prove the case of constant $q \geq 2$. 

\begin{theorem} [Quantum lower bounds for constant $q \geq 2$] \label{thm:lb-constant-q}
    Given quantum query access to the oracle for a probability distribution $p$, for every constant integer $q \geq 2$, any quantum query algorithm that estimates $\mathrm{H}_q\rbra{p}$ to within additive error $\varepsilon$ requires query complexity $\Omega\rbra{\frac{1}{\varepsilon}}$ for sufficiently small $\varepsilon > 0$.
\end{theorem}
\begin{proof}
    We start with considering the problem of distinguishing the two probability distributions $p^{\pm}$ on the sample space $\cbra{0, 1}$, defined by
    \begin{align}
        p^\pm_0 & = \frac{2}{3} \pm \varepsilon, \\
        p^\pm_1 & = \frac{1}{3} \mp \varepsilon,
    \end{align}
    where $\varepsilon > 0$ is sufficiently small. 
    This hard instance was used in \cite[Theorem 5]{CWLY23} and \cite[Lemma 3]{GHYZ24} where they considered the sample complexity of estimating the purity of quantum states with their spectra being the probability distributions $p^{\pm}$, respectively. 
    For constant integer $q \geq 2$, similar to the proof of \cref{thm:lb}, it can be shown that 
    \begin{align}
        \mathrm{H}_q\rbra{p^-} - \mathrm{H}_q\rbra{p^+} & \geq \Omega\rbra{\varepsilon}, \\
        d_{\mathrm{H}}\rbra{p^+, p^-} & \leq O\rbra{\varepsilon}. 
    \end{align}
    Therefore, any estimator for Tsallis entropy to within additive error $\Theta\rbra{\varepsilon}$ can be used to distinguish the two probability distributions $p^\pm$, which, by \cref{thm:qlower-dis-prob}, requires quantum query complexity 
    \begin{equation}
        \Omega\rbra*{\frac{1}{d_{\mathrm{H}}\rbra{p^+, p^-}}} = \Omega\rbra*{\frac{1}{\varepsilon}}.
    \end{equation}
\end{proof}

\section{Lower Bounds for Approximating Monomials} \label{sec:approx-deg}

In this section, we provide a proof of the lower bound $\Omega\rbra{\sqrt{q}}$ on the approximate degree of monomials $x^q$, which is based on the quantum (meta-)algorithm in \cref{lemma:meta} and the information-theoretic quantum query lower bound for Tsallis entropy estimation in \cref{thm:lb}. 

In the following, we first give a (non-uniform) quantum query algorithm for estimating the Tsallis entropy of a quantum state.

\begin{lemma} [(Non-uniform) quantum estimator for Tsallis entropy] \label{lemma:q-complexity-non-uniform}
    Given quantum query access to the state-preparation circuit of a mixed quantum state $\rho$, for every $t \in \rbra{0, 1}$, there is a (non-uniform) quantum query algorithm that, with probability $\geq \frac{2}{3}$, estimates $\mathrm{S}_q\rbra{\rho}$ to within additive error $\varepsilon$ with query complexity 
    \begin{equation}
        O\rbra*{\frac{\widetilde{\deg}_{t\rbra{q-1}\varepsilon}\rbra{x^{q-1}, \sbra{-1, 1}, \sbra{-1, 1}}}{\rbra{1-t}q\varepsilon}}.
    \end{equation}
\end{lemma}

\begin{proof}
Let $p \in \mathbb{R}\sbra{x}$ be the best approximation polynomial for $x^{q-1}$ (note that $p$ is not necessarily computable) such that 
\begin{equation}
    \sup_{x \in \sbra{-1, 1}} \abs*{p\rbra{x} - x^{q-1}} \leq \varepsilon_{\poly}, \quad \deg\rbra{p} = \widetilde{\deg}_{\varepsilon_{\poly}}\rbra{x^{q-1}, \sbra{-1, 1}, \sbra{-1, 1}}.
\end{equation}
Without loss of generality, we can assume that $p$ is even/odd with parity $\rbra{q-1} \bmod 2$ (see \cref{sec:even-odd-poly-approx}). 

By \cref{lemma:meta} with $\varepsilon_{\textup{QSVT}} \coloneqq 0$, $\varepsilon_{\poly} \coloneqq t\rbra{q-1}\varepsilon$ and $\varepsilon_{\textup{QAE}} \coloneqq \rbra{1-t}\rbra{q-1}\varepsilon/2$, there is a (non-uniform) quantum query algorithm that estimates $\mathrm{S}_q\rbra{\rho}$ to within additive error 
\begin{equation}
    \frac{1}{q-1} \rbra*{\varepsilon_{\textup{QSVT}} + \varepsilon_{\poly} + 2\varepsilon_{\textup{QAE}}} = \varepsilon
\end{equation}
and with query complexity
\begin{equation}
    O\rbra*{\frac{\deg\rbra{p}}{\varepsilon_{\textup{QAE}}}} = O\rbra*{\frac{\widetilde{\deg}_{t\rbra{q-1}\varepsilon}\rbra{x^{q-1}, \sbra{-1, 1}, \sbra{-1, 1}}}{\rbra{1-t}q\varepsilon}}.
\end{equation}

\end{proof}

Combining the (non-uniform) quantum query algorithm in \cref{lemma:q-complexity-non-uniform} and the quantum lower bound in \cref{thm:lb}, we can show a matching lower bound on the approximate degree of the monomial $x^q$ as follows. 

\begin{theorem} \label{thm:monomial-approx-lower-bound-fixed-range}
    For every constant $\varepsilon \in \rbra{0, \frac{1}{e}}$, we have 
    \begin{equation}
        \widetilde{\deg}_{\varepsilon}\rbra{x^{q}, \sbra{-1, 1}, \sbra{-1, 1}} = \Omega\rbra{\sqrt{q}}
    \end{equation}
    for all sufficiently large integer $q$. 
\end{theorem}

\begin{proof}
By \cref{lemma:q-complexity-non-uniform}, for every constant $t \in \rbra{0, 1}$, there is a (non-uniform) quantum query algorithm that estimates $\mathrm{S}_q\rbra{\rho}$ to within additive error $\varepsilon$ with probability $\geq \frac{2}{3}$ with query complexity
\begin{equation}
    O\rbra*{\frac{\widetilde{\deg}_{t\rbra{q-1}\varepsilon}\rbra{x^{q-1}, \sbra{-1, 1}, \sbra{-1, 1}}}{q\varepsilon}}.
\end{equation}
By \cref{thm:lb}, for every constant $\gamma \in \rbra{0, 1/e}$ and all sufficiently large integer $q \geq 3$, we have that for every $\varepsilon \in \rbra{0, \gamma/q}$,
\begin{equation}
    O\rbra*{\frac{\widetilde{\deg}_{t\rbra{q-1}\varepsilon}\rbra{x^{q-1}, \sbra{-1, 1}, \sbra{-1, 1}}}{q\varepsilon}} \geq \Omega\rbra*{\frac{1}{\sqrt{q}\varepsilon}},
\end{equation}
which gives
\begin{equation} \label{eq:deg-geq-sqrt-q}
    \widetilde{\deg}_{t\rbra{q-1}\varepsilon}\rbra{x^{q-1}, \sbra{-1, 1}, \sbra{-1, 1}} = \Omega\rbra{\sqrt{q}}.
\end{equation}
By letting $\varepsilon' = t\rbra{q-1}\varepsilon \in \rbra{0, t\gamma\rbra{1-1/q}}$, $t' = t\gamma$ and $q' = q - 1$, we can rephrase \cref{eq:deg-geq-sqrt-q} as follows: 
for every constant $t' \in \rbra{0, 1/e}$ and all sufficiently large integer $q' \geq 3$, 
\begin{equation}
    \widetilde{\deg}_{\varepsilon'}\rbra{x^{q'}, \sbra{-1, 1}, \sbra{-1, 1}} = \Omega\rbra{\sqrt{q'}}
\end{equation}
for every $\varepsilon' \in \rbra{0, t'-t'/q}$.
This actually implies that for every constant $\varepsilon'' \in \rbra{0, 1/e}$, we have 
\begin{equation}
    \widetilde{\deg}_{\varepsilon''}\rbra{x^{q''}, \sbra{-1, 1}, \sbra{-1, 1}} = \Omega\rbra{\sqrt{q''}}
\end{equation}
for sufficiently large integer $q'' \geq 3$, which can be seen by taking constant $t' = \frac{\varepsilon''}{2} + \frac{1}{2e} \in \rbra{\varepsilon'', 1/e}$ and sufficiently large integer $q' > \frac{t'}{t'-\varepsilon''}$. 
\end{proof}

To make the lower bound in \cref{thm:monomial-approx-lower-bound-fixed-range} more general, we extend it to the case where the function value of the approximation polynomial can be any real numbers. 

\begin{theorem} \label{thm:approx-deg-lb}
    For every constant $\varepsilon \in \rbra{0, \frac{1}{2e}}$, we have
    \begin{equation}
        \widetilde{\deg}_{\varepsilon}\rbra{x^{q}, \sbra{-1, 1}, \mathbb{R}} = \Omega\rbra{\sqrt{q}}
    \end{equation}
    for sufficiently large integer $q$.
\end{theorem}

\begin{proof}
Let $p \in \mathbb{R}\sbra{x}$ be the best approximation polynomial of $x^{q}$ such that 
\begin{equation}
    \sup_{x \in \sbra{-1, 1}} \abs*{p\rbra{x} - x^{q}} \leq \varepsilon, \quad \deg\rbra{p} = \widetilde{\deg}_{\varepsilon}\rbra{x^{q}, \sbra{-1, 1}, \mathbb{R}}.
\end{equation}
Let $r\rbra{x} = \rbra{1-\varepsilon}p\rbra{x}$.
Then,
\begin{equation}
    \sup_{x \in \sbra{-1, 1}} \abs*{r\rbra{x} - x^{q}} \leq 2\varepsilon, \quad \sup_{x \in \sbra{-1, 1}} \abs{r\rbra{x}} \leq 1.
\end{equation}
Therefore,
\begin{equation}
    \deg\rbra{p} = \deg\rbra{r} \geq \widetilde{\deg}_{2\varepsilon}\rbra{x^{q}, \sbra{-1, 1}, \sbra{-1, 1}} \geq \Omega\rbra{\sqrt{q}},
\end{equation}
where the last inequality is by \cref{thm:monomial-approx-lower-bound-fixed-range}. 
\end{proof}

\bibliographystyle{alphaurl}
\bibliography{main}

\appendix

\section{Best Polynomial Approximation of Even/Odd Functions} \label{sec:even-odd-poly-approx}

In this appendix, for completeness, we show that any even/odd function always has an even/odd best approximation polynomial with the same parity. 

\begin{lemma}
    For every even (resp.\ odd) function $f \colon \sbra{-a, a} \to \mathbb{R}$ with $a > 0$, if there is a polynomial $p \in \mathbb{R}\sbra{x}$ such that 
    \begin{equation}
        \sup_{x \in \sbra{-a, a}} \abs*{ p\rbra{x} - f\rbra{x} } \leq \varepsilon
    \end{equation}
    for some $\varepsilon > 0$, then there is an even (resp.\ odd) polynomial $p^*$ of degree $\leq \deg\rbra{p}$ such that
    \begin{equation}
        \sup_{x \in \sbra{-a, a}} \abs*{ p^*\rbra{x} - f\rbra{x} } \leq \varepsilon.
    \end{equation}
\end{lemma}
\begin{proof}
    Let
    \begin{equation}
        p_{\textup{even}} = \frac{p\rbra{x} + p\rbra{-x}}{2}, \quad p_{\textup{odd}} = \frac{p\rbra{x} - p\rbra{-x}}{2}.
    \end{equation}
    Then, $p\rbra{x} = p_{\textup{even}}\rbra{x} + p_{\textup{odd}}\rbra{x}$.
    We only consider the case that $f$ is even, and the case that $f$ is odd can be shown similarly. 
    When $f$ is even, then we can choose $p^* = p_{\textup{even}}$. 
    To see this, note that for any 
    $x \in \sbra{-a, a}$, we have
    \begin{align}
        \abs{p\rbra{x} - f\rbra{x}}
        & = \abs{p_{\textup{even}}\rbra{x} + p_{\textup{odd}}\rbra{x} - f\rbra{x}} \\
        & \geq \abs{p_{\textup{even}}\rbra{x} - f\rbra{x}} + \abs{p_{\textup{odd}}\rbra{x}} \\
        & \geq \abs{p_{\textup{even}}\rbra{x} - f\rbra{x}},
    \end{align}
    which gives
    \begin{equation}
        \sup_{x \in \sbra{-a, a}} \abs*{ p_{\textup{even}}\rbra{x} - f\rbra{x} } \leq \sup_{x \in \sbra{-a, a}} \abs*{ p\rbra{x} - f\rbra{x} } \leq \varepsilon.
    \end{equation}
\end{proof}

\section{Sachdeva and Vishnoi's Lower Bound for Approximating Monomials}

In this appendix, for completeness, we reproduce Sachdeva and Vishnoi's proof of the lower bound on the approximate degree of the monomials in \cite{SV14}. 
To this end, we need the following theorems. 

\begin{theorem} [Markov brothers' inequality, adapted from \cite{Mar90}] \label{thm:markov-brother}
    Let $p \in \mathbb{R}\sbra{x}$ be a polynomial of degree $d$ such that $\abs{p\rbra{x}} \leq 1$ for all $x \in \sbra{-1, 1}$.
    Then, the derivative, $p'\rbra{x}$, of $p\rbra{x}$ satisfies
    \begin{equation}
        \abs{p'\rbra{x}} \leq d^2
    \end{equation}
    for all $x \in \sbra{-1, 1}$.
\end{theorem}

\begin{theorem} [Mean value theorem, cf.\ {\cite[Theorem 5.10]{Rud76}}] \label{thm:mean-value}
    If $f \colon \sbra{a, b} \to \mathbb{R}$ is a continuous function and it is differentiable in $\rbra{a, b}$, then there is a point $x \in \rbra{a, b}$ such that 
    \begin{equation}
        f\rbra{b} - f\rbra{a} = \rbra{b-a} f'\rbra{x}.
    \end{equation}
\end{theorem}

Now we are ready to give a detailed proof of the following result that was sketched in \cite[Section 5]{SV14}. 

\begin{theorem} [Lower bound for approximating monomials, \cite{SV14}] \label{thm:approx-deg-SV14}
    For every constant $\varepsilon \in \rbra{0, \frac{e-1}{2e}}$, 
    \begin{equation}
        \widetilde{\deg}_{\varepsilon}\rbra{x^{q}, \sbra{-1, 1}, \sbra{-1, 1}} = \Omega\rbra{\sqrt{q}}.
    \end{equation}
\end{theorem}

\begin{proof}
    Let $p \colon \sbra{-1, 1} \to \sbra{-1, 1}$ be a polynomial of degree $\deg\rbra{p} = \widetilde{\deg}_{\varepsilon}\rbra{x^{q}, \sbra{-1, 1}, \sbra{-1, 1}}$ such that
    \begin{equation}
        \sup_{x \in \sbra{-1, 1}} \abs{p\rbra{x} - x^q} \leq \varepsilon.
    \end{equation}
    Applying the mean value theorem (\cref{thm:mean-value}) with $a = 1-\frac{1}{q}$, $b = 1$, and $f\rbra{x} = p\rbra{x}$, we have that there exists a point $x^* \in \rbra{1-\frac{1}{q}, 1}$ such that
    \begin{align}
        p'\rbra{x^*}
        & = q \rbra*{ p\rbra{1}-p\rbra*{1-\frac{1}{q}} } \\
        & \geq q \rbra*{ \rbra*{\left.x^q\right|_{x = 1} - \varepsilon} - \rbra*{\left.x^q\right|_{x = 1 - \frac 1 q} + \varepsilon}} \\
        & = q \rbra*{ 1 - \rbra*{1-\frac{1}{q}}^q - 2\varepsilon } \\
        & > q \rbra*{1 - \frac{1}{e} - 2\varepsilon}.
    \end{align}
    By Markov brothers' inequality (\cref{thm:markov-brother}), we have
    \begin{align}
        \deg\rbra{p} 
        & \geq \sqrt{\abs{p'\rbra{x^*}}} \\
        & \geq \sqrt{q\rbra*{1 - \frac{1}{e} - 2\varepsilon}}.
    \end{align}
\end{proof}

To make the lower bound in \cref{thm:approx-deg-SV14} more general, we extend it to the case where the function value of the approximation polynomial can be any real numbers. 

\begin{theorem} \label{thm:lb-by-sv14}
    For every constant $\varepsilon \in \rbra{0, \frac{e-1}{4e}}$, 
    \begin{equation}
        \widetilde{\deg}_{\varepsilon}\rbra{x^{q}, \sbra{-1, 1}, \mathbb{R}} = \Omega\rbra{\sqrt{q}}.
    \end{equation}
\end{theorem}
\begin{proof}
    By \cref{thm:approx-deg-SV14} using the same arguments as in \cref{thm:approx-deg-lb}, we have
    \begin{equation}
        \widetilde{\deg}_{\varepsilon}\rbra{x^{q}, \sbra{-1, 1}, \mathbb{R}} \geq \widetilde{\deg}_{2\varepsilon}\rbra{x^{q}, \sbra{-1, 1}, \sbra{-1, 1}} = \Omega\rbra{\sqrt{q}}.
    \end{equation}
\end{proof}

\end{document}